\documentclass[12pt]{article}

\usepackage{amssymb,amsmath,amsfonts,eurosym, geometry,ulem, color,setspace,sectsty, comment,footmisc, pdflscape, array,hyperref}
\usepackage[numbers,square,longnamesfirst]{natbib}
\usepackage{subcaption}
\usepackage{graphicx}
\usepackage{amsmath}
\usepackage{amsthm}
\usepackage{amssymb}
\usepackage{multirow}
\usepackage{url}
\usepackage{authblk}
\usepackage{enumitem}
\usepackage{algorithm}
\usepackage{algorithmic}
\usepackage{booktabs}

\usepackage[percent]{overpic}  

\usepackage{lineno}

\newtheorem{Theorem}{Theorem}

\newcolumntype{L}[1]{>{\raggedright\let\newline\\arraybackslash\hspace{0pt}}m{#1}}
\newcolumntype{C}[1]{>{\centering\let\newline\\arraybackslash\hspace{0pt}}m{#1}}
\newcolumntype{R}[1]{>{\raggedleft\let\newline\\arraybackslash\hspace{0pt}}m{#1}}

\begin{document}

\begin{titlepage}
\title{Adaptive Epidemic Dynamics on Hypergraphs with Group-Level Immunization and Rewiring\thanks{This work is supported in part by the Natural Science Foundation of Chongqing under Grant CSTB2025YITP-QCRCX0007, in part by the Fundamental Research Funds for the Central Universities under Grant SWU-KT26010, and in part by Slovenian Research and Innovation Agency under Grant P1-0403.}}

\author[1]{Yusheng Li}
\author[2]{Meiling Xie}
\author[3]{Qin Li\thanks{Corresponding author. Email: \texttt{qinli1022@swu.edu.cn}}}
\author[1]{Minyu Feng\thanks{Corresponding author. Email: \texttt{myfeng@swu.edu.cn}}}
\author[4]{Matja\v{z} Perc}

\affil[1]{College of Artificial Intelligence, Southwest University, Chongqing 400715, China}
\affil[2]{School of Systems Science/Institute of Nonequilibrium Systems, Beijing Normal University, 100875 Beijing, China}
\affil[3]{Business College, Southwest University, Chongqing 400715, China}
\affil[4]{Faculty of Natural Sciences and Mathematics, University of Maribor, Koro{\v s}ka cesta 160, Maribor, 2000, Slovenia}
\affil[4]{Community Healthcare Center Dr. Adolf Drolc Maribor, Ulica talcev 9, Maribor, 2000, Slovenia}
\affil[4]{Department of Physics, Kyung Hee University, 26 Kyungheedae-ro, Seoul, 02447, Dongdaemun-gu, Republic of Korea}
\affil[4]{University College, Korea University, 145 Anam-ro, Seoul, 02841, Seongbuk-gu, Republic of Korea}

\date{}

\maketitle
\begin{abstract}
\noindent Understanding how higher-order social structures shape epidemic spreading requires models that couple group interactions with adaptive behavior. We introduce an adaptive simplicial susceptible–infected–susceptible ($s$-SIS) model on $d$-uniform hypergraphs, where both node states and hyperedge activity co-evolve in response to local infection pressure. Hyperedges represent group interactions of fixed size and dynamically reduce their activity through a feedback mechanism in highly infected environments. Within this framework, we design two classes of hyperedge-level interventions: (i) risk-driven immunization, combining spontaneous, activity-based isolation with targeted deactivation guided by hyperedge infection pressure, and (ii) structural rewiring, which reconstructs group structures either randomly or via degree-preferential attachment. By extending the microscopic Markov chain approximation to higher-order interactions, we derive analytical conditions for the existence and stability of both endemic and disease-free stationary states. Our analysis shows that adaptive hyperedge feedback can induce discontinuous phase transitions, nonlinear epidemic thresholds, and bistable regimes in which sufficiently high initial prevalence drives the system to a disease-free equilibrium. Extensive Monte Carlo simulations support the theory and confirm that targeted immunization and degree-preferential rewiring substantially suppress epidemic prevalence, outperforming random strategies. These results demonstrate that higher-order interactions and adaptive group-level responses fundamentally reshape epidemic bifurcations and suggest principles for designing effective intervention policies in complex social systems.
\vspace{0in}
\bigskip
\end{abstract}
\setcounter{page}{0}
\thispagestyle{empty}
\end{titlepage}
\pagebreak \newpage

\doublespacing

\section{Introduction}

Modeling infectious disease transmission has long been a central topic in mathematical epidemiology and complex systems science. 
Early compartmental models~\cite{anderson1992infectious,vespignani2012modelling}, based on homogeneous mixing assumptions, employed systems of differential equations to describe the time evolution of classical susceptible–infected–recovered (SIR) and susceptible–infected–susceptible (SIS) dynamics~\cite{tang2025sis,li2022network}. 
These frameworks laid the mathematical foundations for key epidemiological indicators such as the basic reproduction number $R_0$~\cite{xie2023contact} and have guided public health interventions ranging from smallpox eradication to seasonal influenza control~\cite{ferguson2003planning,casagrandi2006sirc}. 
However, the assumption of uniform mixing fails to capture the heterogeneous contact structures of real populations, where interaction networks shape epidemic outcomes. 
	
The emergence of network epidemiology~\cite{pastor2015epidemic} revealed that network topology profoundly modulates disease dynamics. 
Short path lengths in small-world networks accelerate transmission~\cite{watts1998collective}, while highly connected hubs in scale-free networks can sustain infection even at infinitesimal transmission rates~\cite{barabasi1999emergence,pastor2001epidemic}. 
Subsequent advances incorporated multilayer~\cite{boccaletti2014structure,feng2023impact} and metapopulation~\cite{colizza2007modeling,gomez2018critical} frameworks, enabling the modeling of spatial heterogeneity, mobility patterns, and interdependent spreading processes~\cite{soriano2018spreading}. 
These developments established a structural foundation for understanding how network organization controls epidemic thresholds, prevalence levels, and intervention efficiency. 
	
Most network-based models originally assumed fixed topologies, ignoring the feedback between individual behavior and epidemic progression. 
In reality, humans adapt their social interactions dynamically in response to perceived risk, information awareness, or behavioral constraints. 
This coevolutionary perspective gave rise to adaptive networks~\cite{gross2008adaptive, Feng10521680}, where link formation and removal depend on the evolving disease state of nodes~\cite{jolad2012epidemic}. 
Such feedback can induce complex dynamical phenomena absent in static structures, including bistability, hysteresis, and self-organized oscillations~\cite{gross2006epidemic}. 
For example, the rewiring of high-risk contacts has been shown to delay or even suppress outbreaks~\cite{berner2023adaptive}, but excessive adaptation can also fragment the network, altering percolation and resilience properties~\cite{ally2018effects}. 
These findings underscore that epidemic control cannot be fully understood without accounting for the joint evolution of topology and dynamics.

Network structure fundamentally determines epidemic controllability, motivating extensive studies of immunization and intervention strategies~\cite{jimenez2003epidemic,pastor2002epidemic}. 
At the microscopic level, individuals may spontaneously self-isolate or reduce contact frequency, while organized interventions at the macroscopic level (e.g., targeted vaccination or contact restrictions) remove nodes or edges to interrupt transmission pathways.~\cite{pastor2002immunization}. 
The fraction of immunized nodes or edges required to eliminate endemic persistence defines the Herd Immunity Threshold (HIT)~\cite{john2000herd}.
Due to the inefficiency of random immunization in heterogeneous networks, structural information has been leveraged to develop targeted methods, including degree-based vaccination~\cite{holme2002attack}, acquaintance immunization~\cite{cohen2003efficient}, and influence-based strategies~\cite{chen2008finding}. 
Beyond node-level interventions, recent studies have demonstrated that selectively removing or immunizing edges based on network topology can substantially enhance containment efficiency. 
For instance, Liu et al.~\cite{liu2024diffusion} proposed an explosive percolation-based edge removal strategy that selectively disconnects low-degree contacts to curb misinformation spread, while a subsequent work~\cite{liu2026efficient} combined greedy optimization with topological features to develop computationally scalable edge immunization schemes, showing that jointly optimizing local connectivity patterns and global structure yields superior performance. 
Recent work further emphasizes group-level targeting in realistic contact contexts such as households, schools, and workplaces, highlighting the need for intervention models that operate beyond pairwise links~\cite{jhun2021effective}.
In parallel, the microscopic Markov chain approach has been applied to analyze multilayer behavioral contagion~\cite{MB2025green}.

Traditional network representations capture only pairwise interactions, yet real contagion processes often unfold within groups, such as meetings, classrooms, or social gatherings, where multiple individuals interact simultaneously~\cite{battiston2020networks,wang2024epidemic}. 
Hypergraphs and simplicial complexes provide natural mathematical formalisms for such higher-order structures~\cite{courtney2016generalized}. 
In the $d$-uniform hypergraph, each hyperedge connects exactly $d$ nodes, with $d=2$ reducing to a standard network. 
Recent studies employing the simplicial SIS ($s$-SIS) model~\cite{jhun2019simplicial,wang2021simplicial} have demonstrated that group interactions introduce nonlinear reinforcement effects. These effects fundamentally alter epidemic dynamics, giving rise to discontinuous phase transitions, bistability, and hysteresis phenomena~\cite{matamalas2020abrupt,iacopini2019simplicial}. 
Moreover, higher-order coupling lowers effective epidemic thresholds~\cite{deArruda2020,yuan2025impacts}, implying that neglecting group effects can severely underestimate outbreak risks. 
Beyond static higher-order descriptions, recent studies have extended simplicial epidemic models to temporal networks and spatial settings, incorporating control strategies based on critical simplex identification~\cite{PhysRevE.112.044409} and analyzing pattern dynamics through reaction–diffusion frameworks on higher-order topologies~\cite{BMB2025,MB2025pattern}.
However, real-world group interactions are neither static nor continuous, as they vary in frequency and intensity, often adapting to behavioral or environmental feedback.

Despite progress in adaptive networks and higher-order epidemic modeling, a unified framework integrating adaptive group interactions and strategic interventions remains limited. 
Most existing works treat adaptation as either individual-level rewiring or static group suppression, and even those that consider hyperedge removal typically rely on fixed structural rankings rather than dynamically updated risk measures. 
The intervention therefore does not respond to the shifting landscape of transmission risk during an ongoing outbreak. 
To bridge this gap, we develop an adaptive $s$-SIS model on $d$-uniform hypergraphs, where the evolution of node states and hyperedge activity are mutually coupled. 
Within this framework, we introduce two classes of hyperedge-level interventions: 
(i) Risk-based immunization, including spontaneous isolation and external deactivation; 
(ii) Structural rewiring, implemented through random or degree-preferential strategies that reshape the topology of the network. 
We extend the microscopic Markov chain approximation (MMCA)~\cite{gomez2010discrete} to analyze this model, deriving analytical conditions for stationary states and nonlinear epidemic thresholds. 
The theoretical predictions are further validated by Monte Carlo (MC) simulations, allowing us to quantify how the interplay of topology, behavior, and interventions governs epidemic resilience.

To summarize, the principal contributions of our work are as follows.
\begin{enumerate}[label=\arabic*), itemsep=3pt, parsep=0pt]
    \item We propose and formalize an adaptive $s$-SIS model on uniform hypergraphs, establishing a unified framework for higher-order epidemic dynamics with feedback-driven hyperedge activity.
    \item We design and systematically compare multiple hyperedge-level immunization and rewiring strategies, quantifying their effects on the stationary-state infection density and network structure.
    \item We demonstrate the emergence of discontinuous phase transitions, nonlinear thresholds, and hysteresis phenomena through analytical derivations and extensive simulations, revealing how adaptive feedback and initial conditions jointly shape epidemic outcomes.
\end{enumerate}

The remainder of this paper is structured as follows.
Section~\ref{sec:model} develops the adaptive simplicial SIS model, integrating both intervention strategies and rewiring mechanisms within a unified adaptive framework, and presents the corresponding MMCA analytical formulation.
Section~\ref{sec:simulations} conducts numerical simulations and Monte Carlo validations.
Finally, in Section~\ref{sec:conclusion}, we summarize the conclusions and discuss their broader implications.

\section{Modelling}		
\label{sec:model}

This section introduces the adaptive simplicial SIS model on $d$-uniform hypergraphs and its associated hyperedge-level intervention mechanisms. 
We first describe the adaptive structure that couples node-level epidemic dynamics with hyperedge activity in Section~\ref{model_des}, 
followed by the formulation of two classes of intervention strategies, immunization and rewiring, which act at the group level to mitigate epidemic spread in Section~\ref{sec:intervention}. 
Finally, in Section~\ref{theory}, we provide a theoretical analysis of the model using a homogeneous MMCA framework to derive nonlinear thresholds and equilibrium conditions. Table~\ref{tab:notations} summarizes the notations used throughout this paper.
	
\subsection{Simplicial SIS model with adaptive hyperedges}  \label{model_des}
	
We establish an adaptive epidemic modelling framework that explicitly integrates the coevolution between disease transmission and higher-order social interactions. 
In many real-world systems, such as group gatherings, collaborative teams, or biological complexes, contagion does not spread solely through pairwise contacts but through group-based interactions involving multiple individuals simultaneously. 
To capture these higher-order effects, we represent the population as a $d$-uniform hypergraph, which naturally extends classical pairwise networks by allowing each hyperedge to connect $d$ distinct nodes instead of two.

\begin{table}[!htbp]
\centering
\caption{Summary of key variables and parameters in the adaptive $s$-SIS model.}
\label{tab:notations}
\begin{tabular}{>{\centering}m{1.5cm} >{\centering}m{4cm} >{\raggedright\arraybackslash}m{12cm}}
\toprule
\multicolumn{1}{c}{\textbf{Symbol}} &
\multicolumn{1}{c}{\textbf{Category}} &
\multicolumn{1}{c}{\textbf{Description}} \\
\midrule

$n$ & Network indices & Total number of nodes in the system. \\
$m$ & Network indices & Total number of hyperedges. \\
$d$ & Network indices & Size of each hyperedge in the $d$-uniform hypergraph. \\
$H(n,m,d)$ & Network indices & Random $d$-uniform hypergraph representation. \\
$\mathcal{V},\ \mathcal{E}$ & Network indices & Sets of vertices and hyperedges, respectively. \\
$\mathcal{E}_i$ & Network indices & Set of hyperedges incident to node $i$. \\
$k_i$ & Network indices & Hyperdegree of node $i$. \\
$\langle k\rangle$ & Network indices & Average hyperdegree across all nodes. \\
$X_i(t)$ & Network indices & Disease state of node $i$ at time $t$ ($S$ or $I$). \\
$p_i(t)$ & Model parameters & Probability that node $i$ is infected at time $t$. \\
$\beta_d$ & Model parameters & Disease transmission probability through a hyperedge, the same as $\beta$ \\
$\mu$ & Model parameters & Recovery rate from infected to susceptible state. \\
$\Phi_e(t)$ & Model parameters & Infection pressure of hyperedge $e$, defined as the joint infection probability of all its member nodes. \\
$\Theta_e(t)$ & Model parameters & Activity level of hyperedge $e$, ranging from $0$ to $1$. \\
$\eta$ & Model parameters & Sensitivity parameter controlling how strongly activity responds to local infection pressure. \\
$\gamma$ & Model parameters & Spontaneous recovery rate of an inactive hyperedge toward full activity. \\
$q_i(t)$ & Model parameters & Probability that susceptible node $i$ avoids infection from all its incident hyperedges. \\

$w$ & Intervention factors & Fraction of hyperedges that can be immunized under a given resource budget. \\

$\theta_{\mathrm{min}}$ & Intervention factors & Critical activity threshold triggering spontaneous isolation of a hyperedge. \\
$\mathcal{U}_d$ & Intervention factors & Set of all possible $d$-node subsets, used for random rewiring. \\
$\Theta_0$ & Intervention factors & Initial activity level assigned to newly rewired hyperedges. \\
$\alpha$ & Intervention factors & Offset parameter ensuring non‑zero sampling probability in degree‑preferential rewiring. \\

$\rho$ & Theoretical analysis & Stationary mean infection density at equilibrium. \\
$\rho(t)$ & Theoretical analysis & Global infection prevalence in the system. \\
$\Theta^{\star}$ & Theoretical analysis & Stationary mean hyperedge activity at equilibrium. \\
$\beta_n$ & Theoretical analysis & Nonlinear epidemic threshold for higher‑order interactions ($d \geq 3$). \\
$\beta_c$ & Theoretical analysis & Classical linear epidemic threshold for pairwise interactions ($d = 2$). \\
$G(\rho)$ & Theoretical analysis & Self‑consistency function used in the homogeneous MMCA framework. \\
$\Lambda_{\max}(\mathbf{M})$ & Theoretical analysis & Leading eigenvalue of the adjacency matrix, used in the linear threshold condition. \\

\bottomrule
\end{tabular}
\end{table}

Formally, we generalize the Erdős–Rényi random graph $G(n,m)$ to its higher-order analogue $H(n,m,d)$, 
where $n$ denotes the number of nodes (individuals), $m$ the number of hyperedges (group interactions), and $d$ the hyperedge dimension representing group size. 
The vertex set is defined as $\mathcal{V}=\{v_1,v_2,\dots,v_n\}$, and the hyperedge set as $\mathcal{E}=\{e_1,e_2,\dots,e_m\}$, where each hyperedge $e\subseteq\mathcal{V}$ satisfies $|e|=d$. 
In this construction, all hyperedges are treated as statistically equivalent, implying uniform group sizes and random composition of members.

In the sparse regime where $m\ll n^d$, the probability that two nodes share multiple hyperedges becomes negligible, 
and the hyperdegree distribution referred to the number of hyperedges incident to a node is approximately Poisson with mean
\begin{equation}
\langle k \rangle = \frac{md}{n} = \frac{1}{n}\sum_{i=1}^{n} |\mathcal{E}_i|,
\end{equation}
where $\mathcal{E}_i$ denotes the set of hyperedges containing node $i$. 
The average hyperdegree $\langle k\rangle$ characterizes the expected number of group interactions per individual and thus determines the effective connectivity of the system.

In real-world social and biological systems, contagion processes are rarely independent of the network structure on which they occur. 
Individuals tend to alter their group participation in response to perceived infection risk, while the collective organization of social groups evolves as a consequence of behavioral adaptation. 
Such adaptive feedback between disease dynamics and network structure has been widely observed during epidemics, where people spontaneously reduce high-risk interactions but gradually restore connectivity once perceived danger declines. 
To capture this adaptive interplay, our model introduces a dynamical mechanism in which epidemic dynamics and hyperedge activity evolve simultaneously.

Specifically, the model integrates two coupled processes, as illustrated in Fig.~\ref{fig:adaptiveframework}(a):

\begin{enumerate}[label=\arabic*), itemsep=2pt, parsep=0pt, topsep=4pt]
  \item \textbf{Node-level epidemic dynamics:} Each node $i$ exhibits a binary state $X_i(t) \in \{0,1\}$, representing susceptible ($S$) and infected ($I$) states. 
  Infection transmission occurs through higher-order interactions, requiring a critical configuration in which all other $d-1$ members within a hyperedge are infected. 
  Each such hyperedge independently transmits infection with probability $\beta_d$ per time step, while infected nodes recover with probability $\mu$.

  \item \textbf{Hyperedge-level adaptive dynamics:} Each hyperedge $e$ is endowed with a \textit{dynamic activity level} $\Theta_e(t) \in [0,1]$, representing its effective interaction intensity. 
  This variable captures behavioral adaptation at the group level: hyperedges with high infection pressure tend to reduce activity, while inactive ones gradually regain it over time. 
  The temporal evolution of $\Theta_e(t)$ follows
  \begin{equation}
  \label{eq:Theta_update}
  \Theta_e(t+1) = \Theta_e(t)\left[1-\eta\Phi_e(t)\right] + \gamma\left[1-\Theta_e(t)\right],
  \end{equation}
  where $\eta \in [0,1]$ quantifies sensitivity to local infection pressure, and $\gamma \in [0,1]$ characterizes spontaneous recovery of interaction activity.
  The first term drives the activity downward when the hyperedge is exposed to high infection pressure $\Phi_e(t)$, while the second term drives a spontaneous return toward full activity irrespective of the current state.
  Together they produce a dynamic balance between risk-driven suppression and intrinsic recovery.
\end{enumerate}

The \emph{hyperedge infection pressure} $\Phi_e(t)$, defined as the joint probability that all nodes in $e$ are simultaneously infected, is given by
\begin{equation}
\label{eq:Phi}
\Phi_e(t) = \Pr ( \bigcap_{i \in e} \{X_i(t) = I\}) \approx \prod_{i \in e} p_i(t) ,
\end{equation} 
where $p_i(t) = \Pr(X_i(t) = I)$.
Equation~(3) extends the edge epidemic importance (EI)-based immunization strategy introduced in \cite{matamalas2018effective}. It computes the joint probability that all nodes within a hyperedge are simultaneously infected, which serves as the basis for ranking and deactivating or rewiring hyperedges in our intervention strategies. This probability follows from the individual-based mean-field (IBMF) theory on hypergraphs, which tracks each node's infection probability $p_i$ independently and neglects dynamical correlations between node states \cite{gomez2010discrete}. Hence the approximation $\Phi_e(t) \approx \prod_{i\in e} p_i(t)$ is a natural consequence of the IBMF closure.

The probability that susceptible node $i$ avoids infection through all incident hyperedges is
\begin{equation}
\label{eq:qi(t)}
q_i(t) = \prod_{e\in\mathcal{E}_i}\left(1-\Theta_e(t)\beta_d\prod_{\substack{j\in e \\ j\neq i}} p_j(t)\right),
\end{equation}
leading to the complete node-state evolution as
\begin{equation}
\label{eq:mmca_node}
p_i(t+1) = \left(1-p_i(t)\right)\left(1-q_i(t)\right) + (1-\mu)p_i(t).
\end{equation}

Eq.~\eqref{eq:Theta_update}--\eqref{eq:mmca_node} define the microscopic dynamics. 
They capture the core feedback loop where increased local infection pressure reduces hyperedge activity introduced in Eq.~\eqref{eq:Theta_update}, which in turn lowers transmission probabilities in subsequent steps in Eq.~\eqref{eq:qi(t)}, and vice versa.

For analytical tractability, the model relies on several simplifying assumptions:
(i) the hypergraph is $d$-uniform and constructed uniformly at random, so that hyperedges are statistically equivalent and the hyperdegree distribution is approximately Poisson;
(ii) all hyperedges share the same adaptation parameters $\eta$ and $\gamma$, implying uniform responsiveness to infection pressure across groups;
(iii) node states within a hyperedge are treated as independent.

\begin{figure*}[!t]
  \centering
  \vspace{-4em}
  \includegraphics [width=0.95\textwidth]{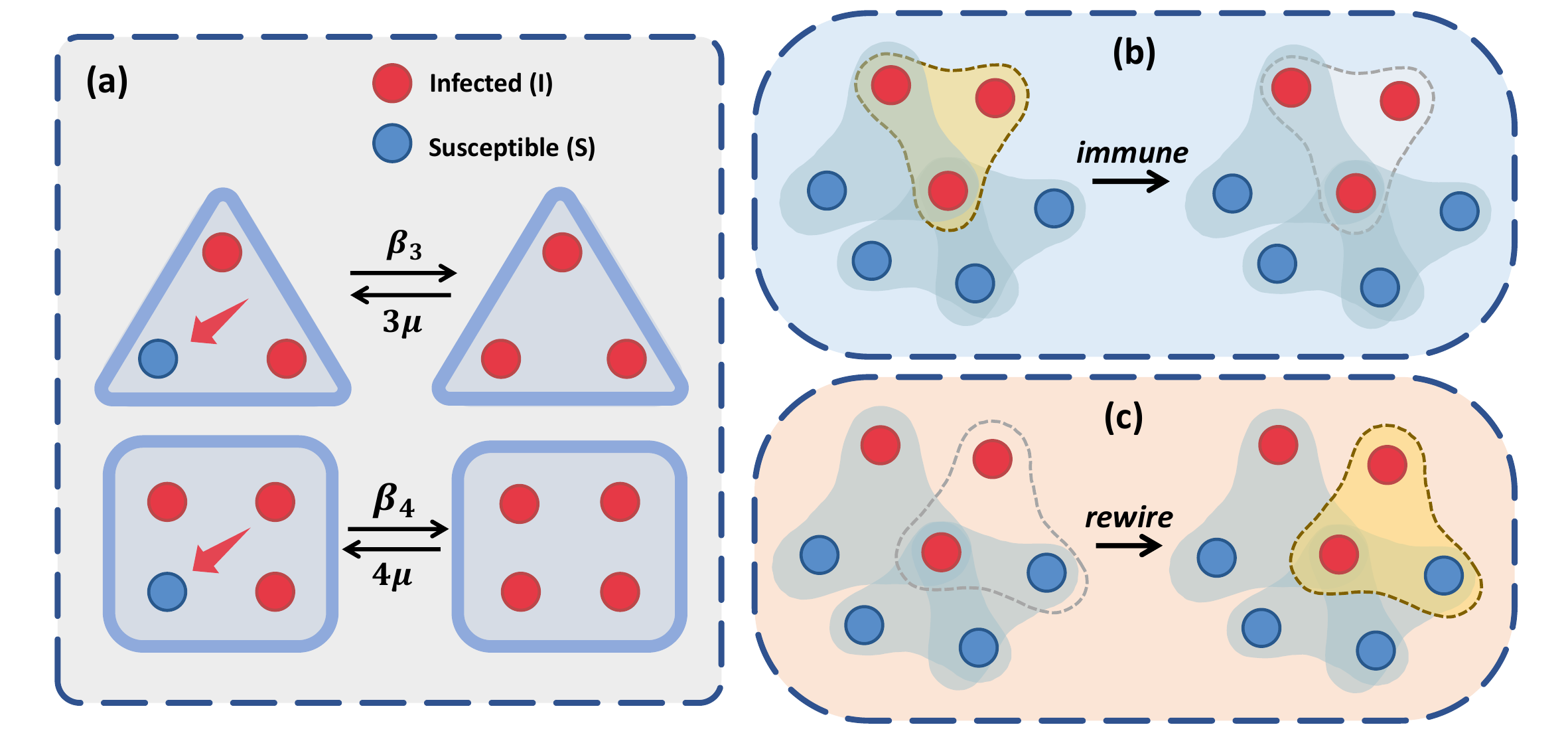}
  \caption{\textbf{Schematic illustration of the adaptive epidemic model on uniform hypergraphs.} 
   (a) Demonstration of the $s$-SIS model. 
   Red and blue circles represent infected and susceptible nodes, respectively, connected by hyperedges (light blue regions). 
   Infection spreads with rate $\beta_d$ when all $d-1$ other members in a hyperedge are infected, 
   whereas each infected node recovers independently at rate $\mu$. 
   (b) Illustration of the immunization process on a 3-uniform hypergraph. 
   Hyperedges experiencing high infection pressure are selectively deactivated according to the adopted intervention strategy.
   (c) Illustration of the rewiring mechanism. 
   Deactivated hyperedges are replaced by newly sampled ones to preserve structural connectivity while modifying higher-order interaction patterns. 
   Together, these processes create a closed feedback loop between node-level epidemic states and the adaptive evolution of hyperedge activity.}
  \label{fig:adaptiveframework}
\end{figure*}

\subsection{Hyperedge-level intervention strategies}  \label{sec:intervention}
	
\begin{algorithm}[t]
\caption{Spontaneous Isolation}
\label{alg:spontaneous-isolation}
\begin{algorithmic}[1]
\REQUIRE Hyperedge set $\mathcal{E}$, activity levels $\{\Theta_e\}$, threshold $\theta_{\mathrm{min}}$, budget $w$
\ENSURE Immunized hyperedge set $\mathcal{E}_{\text{immune}}$
\STATE $m \leftarrow |\mathcal{E}|$, $m_{\text{budget}} \leftarrow \lfloor w \cdot m \rfloor$
\STATE $\mathcal{E}_{\text{candidates}} \leftarrow \{ e \in \mathcal{E} : \text{active}[e] \text{ and } \Theta_e < \theta_{\mathrm{min}} \}$
\IF{$\mathcal{E}_{\text{candidates}} \neq \emptyset$}
    \STATE Sort $\mathcal{E}_{\text{candidates}}$ by $\Theta_e$ in descending order
    \STATE $\mathcal{E}_{\text{immune}} \leftarrow$ first $\min(m_{\text{budget}}, |\mathcal{E}_{\text{candidates}}|)$ hyperedges from sorted $\mathcal{E}_{\text{candidates}}$
\ELSE
    \STATE $\mathcal{E}_{\text{immune}} \leftarrow \emptyset$
\ENDIF
\FOR{each $e \in \mathcal{E}_{\text{immune}}$}
    \STATE $\text{active}[e] \leftarrow \text{False}$
    \STATE $\Theta_e \leftarrow 0$
\ENDFOR
\RETURN $\mathcal{E}_{\text{immune}}$
\end{algorithmic}
\end{algorithm}

\begin{algorithm}[t]
\caption{Targeted Immunization}
\label{alg:targeted-immunization}
\begin{algorithmic}[1]
\REQUIRE Hyperedge set $\mathcal{E}$, infection pressures $\{\Phi_e\}$, budget $w$
\ENSURE Immunized hyperedge set $\mathcal{E}_{\text{immune}}$
\STATE $m \leftarrow |\mathcal{E}|$, $m_{\text{budget}} \leftarrow \lfloor w \cdot m \rfloor$
\STATE $\mathcal{E}_{\text{active}} \leftarrow \{ e \in \mathcal{E} : \text{active}[e] \}$
\IF{$|\mathcal{E}_{\text{active}}| \leq m_{\text{budget}}$}
    \STATE $\mathcal{E}_{\text{immune}} \leftarrow \mathcal{E}_{\text{active}}$
\ELSE
    \STATE Sort $\mathcal{E}_{\text{active}}$ by $\Phi_e$ in descending order
    \STATE $\mathcal{E}_{\text{immune}} \leftarrow$ first $m_{\text{budget}}$ hyperedges from sorted $\mathcal{E}_{\text{active}}$
\ENDIF
\FOR{each $e \in \mathcal{E}_{\text{immune}}$}
    \STATE $\text{active}[e] \leftarrow \text{False}$
    \STATE $\Theta_e \leftarrow 0$
\ENDFOR
\RETURN $\mathcal{E}_{\text{immune}}$
\end{algorithmic}
\end{algorithm}

In realistic social systems, both individual states and group interactions evolve dynamically in response to epidemic risk. 
To capture these adaptive processes, we integrate intervention strategies directly into the model framework, allowing the hypergraph structure to reorganize under infection pressure. 
Two complementary classes of interventions are considered: 
(i) Immunization, which reduces transmission by deactivating highly infectious hyperedges illustrated in Fig.~\ref{fig:adaptiveframework}(b), and 
(ii) Rewiring, which reconstructs inactive hyperedges to preserve global connectivity shown in Fig.~\ref{fig:adaptiveframework}(c). 
These mechanisms jointly define the adaptive feedback governing higher-order epidemic evolution.

\subsubsection{Immunization strategies}
Immunization aims to suppress contagion by selectively deactivating risky hyperedges, thereby reducing the effective transmission channels without completely dismantling the network. 
To capture both decentralized behavioral responses and centralized policy-driven interventions, we design two complementary strategies and describe their key principles and corresponding mathematical formulations below.

\textit{(1) Spontaneous Isolation (SI).} 
This mechanism mimics decentralized self-protection behavior commonly observed in real epidemics, where social groups spontaneously reduce interaction when perceiving elevated infection risk. 
As illustrated in Algorithm~\ref{alg:spontaneous-isolation}, each hyperedge $e$ has an activity level $\Theta_e(t)$ that reflects its current interaction intensity. 
At every discrete time step, the system evaluates all hyperedges and identifies those whose activity falls below a critical threshold $\theta_{\mathrm{min}}$ as
\begin{equation}
\mathcal{E}_{\text{isolate}}(t) = \left\{ e \in \mathcal{E} : \Theta_e(t) < \theta_{\mathrm{min}}  \right\}.
\label{eq:isolate_set}
\end{equation}
These hyperedges are considered likely to suspend contact spontaneously.

However, to model the realistic limitation of resources and compliance, the total number of isolated hyperedges cannot exceed a global immunization budget $w \in [0,1]$, corresponding to 
$m_{\text{max}} = \lfloor w \cdot m \rfloor$ eligible hyperedges. 
Operationally, the budget $w$ reflects the fraction of hyperedges that can be simultaneously targeted under limited public health capacity, such as the number of available response teams or isolation facilities, and the same uniform budget constraint applies to every hyperedge. 
Beyond constraining the intervention scale, $w$ also serves as the control parameter for identifying the HIT. 

Algorithmically, all candidate hyperedges in $\mathcal{E}_{\text{isolate}}(t)$ are ranked by their current activity level $\Theta_e(t)$ in descending order, and the top $m_{\text{max}}$ are selected for isolation. 
Formally, for each isolated hyperedge $e$, its activity and effective transmission capacity are set to zero, expressed as
\begin{equation}
\Theta_e(t+1) = 0, \qquad \text{active}[e] \leftarrow \text{False},
\end{equation}
ensuring that deactivated hyperedges no longer participate in subsequent infection or rewiring events, effectively removing them from the propagation process. 
Such local, threshold-based decision rules capture spontaneous behavioral adaptation without requiring any global information about network state.

\textit{(2) Targeted Immunization (TI).} 
Unlike SI, the targeted immunization scheme represents a centralized, information-driven strategy designed to allocate limited intervention resources most efficiently. 
Its logic parallels optimal control approaches, where the system identifies and deactivates hyperedges that contribute most to sustaining infection. 
As illustrated in Algorithm~\ref{alg:targeted-immunization}, we employ the hyperedge infection pressure $\Phi_e(t)$ introduced in Eq.~(\ref{eq:Phi}), defined as the joint probability that all nodes within $e$ are infected by 
\begin{equation}
\Phi_e(t) = \prod_{i\in e} p_i(t),
\end{equation}
where the approximation follows from statistical independence under the MMCA framework. 
This criterion extends the edge epidemic importance measure from network-based immunization to higher-order systems. 
Unlike previous hyperedge removal strategies that rank group interactions by static structural properties, the infection pressure $\Phi_e(t)$ is recomputed at every time step as the node infection probabilities evolve, so the set of targeted hyperedges adapts to the real-time epidemic situation rather than being fixed a priori. 
Intuitively, $\Phi_e(t)$ serves as a quantitative measure of the overall infection burden within a hyperedge, establishing itself as a principled risk metric for group-level transmission. 

The algorithm ranks all active hyperedges in decreasing order of their infection pressure $\Phi_e(t)$ and selects the top $m_{\text{target}} = \lfloor w \cdot m \rfloor$ hyperedges for immunization.  
Formally, let $\varepsilon \subseteq \mathcal{E}$ denote a candidate subset of hyperedges with cardinality $|\varepsilon| = m_{\text{target}}$.  
The optimal target set is determined by maximizing the cumulative infection pressure over all such subsets,
\begin{equation}
\mathcal{E}_{\text{target}} = 
\arg\max_{\varepsilon \subseteq \mathcal{E}}
\sum_{e \in \varepsilon} \Phi_e(t).
\end{equation}
Each selected hyperedge $e \in \mathcal{E}_{\text{target}}$ is subsequently deactivated by setting its activity $\Theta_e(t+1)=0$ and flagging it as inactive, consistent with the procedure used in spontaneous isolation.  
This approach concentrates intervention effort on the most influential hyperedges, which act as key transmission hubs, thus achieving a larger reduction in global infection density under the same resource constraints.

\subsubsection{Rewiring mechanism}
While both SI and TI suppress transmission by removing high-risk connections, such removal inevitably alters the network’s connectivity pattern and may lead to fragmentation. 
In real social systems, however, individuals rarely remain permanently disconnected. Instead, they form new groups or reorganize their interactions after previous connections are disrupted. 
To capture this inherent adaptive behavior, we introduce a rewiring process that regenerates new hyperedges in place of the immunized ones. 
This mechanism maintains the total number of hyperedges $m$ constant, ensuring the system preserves its size while allowing the higher-order structure to evolve dynamically. 
Hence, immunization and rewiring jointly constitute a feedback loop—one that suppresses infection through selective deactivation while simultaneously restoring connectivity through adaptive reorganization.

\textit{(1) Random Rewiring.} 
This mechanism represents unbiased social mixing, where each new hyperedge is sampled uniformly at random from the set of all possible $d$-node combinations of $\mathcal{V}$:
\begin{equation}
P\left \{e_{\mathrm{new}} = e\right \}  = \frac{1}{|\mathcal{U}_d|}, \qquad \forall\, e \in \mathcal{U}_d,
\end{equation}
where $\mathcal{U}_d$ denotes the collection of all distinct $d$-element subsets of the node set. 
To model cautious behavioral recovery, all new hyperedges are initialized with a small baseline activity $\Theta_0 = 0.1$, reflecting the gradual restoration of social contacts after isolation.

\textit{(2) Degree-Preferential Rewiring.} 
In many social systems, individuals with higher connectivity are more likely to form new groups after disruptions. 
To incorporate this empirically observed bias, we define the probability that node $i$ is selected during rewiring as proportional to its current hyperdegree $k_i$, adjusted by a small offset $\alpha>0$ to prevent isolation of low-degree nodes as
\begin{equation}
P_{\text{pref}}(i) = \frac{k_i + \alpha}{\sum_{j=1}^n (k_j + \alpha)} = \frac{k_i + \alpha}{md + n\alpha}.
\end{equation}
Newly generated hyperedges under this rule are also initialized with $\Theta_0=0.1$, ensuring consistency across different rewiring strategies.

It is worth noting that the proposed intervention strategies differ substantially in their information requirements and operational practicality.
Spontaneous isolation (SI) relies only on local threshold judgments, since each hyperedge simply compares its own activity $\Theta_e$ against $\theta_{\min}$, and can therefore be implemented in a decentralized manner with minimal information cost.
Targeted immunization (TI), by contrast, requires global sorting of all active hyperedges by their infection pressure $\Phi_e$, demanding centralized surveillance and coordination.
Similarly, random rewiring requires no network information, whereas degree-preferential rewiring needs the current hyperdegrees of all nodes.
These differences imply a trade-off between efficacy and practicality, since strategies that use more information tend to perform better but incur higher implementation costs.

Together, these two intervention components constitute a closed adaptive control loop linking microscopic behavioral responses to macroscopic network resilience. 
The former dynamically suppresses risky interactions, while the latter continuously reorganizes higher-order connectivity, yielding a flexible framework for analyzing the coevolution of epidemic spreading and social adaptation.
In this framework, each new hyperedge is sampled uniformly from the set of all distinct $d$-node subsets, structurally excluding self-loops and repeated hyperedges.
Moreover, when a hyperedge is deactivated, either through adaptive feedback or external intervention, its member nodes remain fully active in all other hyperedges they belong to.
The deactivation only removes the transmission pathway associated with that particular group, reflecting a targeted, partial disruption of social contacts rather than a global isolation of the involved individuals.

\subsection{Theoretical analysis}  
\label{theory}
To gain analytical insight into the adaptive dynamics of the simplicial SIS model, we perform a homogeneous reduction of the MMCA equations, treating nodes and hyperedges as statistically equivalent. 
This scalar approximation replaces microscopic variables with global averages, yielding closed-form self-consistency relations for the mean infection density and average hyperedge activity, which are convenient for threshold and bifurcation analysis.

\medskip
\noindent
\begin{Theorem}
\label{theory1}
For the adaptive $s$-SIS model on a $d$-uniform hypergraph with $d \geq 3$, the disease-free equilibrium $\rho=0$ is locally stable for all $\beta > 0$, and endemic equilibria emerge via a saddle-node bifurcation at the nonlinear threshold $\beta_n$ defined implicitly by the parametric relation
\begin{equation}
\beta_n = \frac{1 - \left( 1 - \frac{\mu \rho}{1-\rho} \right)^{1/\langle k\rangle}}{\Theta^{\star} \, \rho^{d-1}},
\end{equation}
where $\Theta^{\star} = \gamma / (\gamma + \eta \rho^d)$.
\end{Theorem}

\medskip
\noindent
\textbf{Proof.}  
Let $\rho(t)$ denote the average infection probability of nodes at time $t$ and $\Theta(t)$ the mean hyperedge activity:
\begin{equation}
\rho(t) = \frac{1}{n}\sum_{i=1}^{n} p_i(t), \qquad 
\Theta(t) = \frac{1}{m}\sum_{e \in \mathcal{E}} \Theta_e(t).
\end{equation}

Under the homogeneous approximation, each node participates in $\langle k \rangle$ hyperedges, and the average probability that a susceptible node avoids infection from all its incident hyperedges according to Eq.~\eqref{eq:qi(t)} becomes
\begin{equation}
q(t) = \left [  1 - \Theta(t)\, \beta\, \rho(t)^{\,d-1} \right ]^{\langle k \rangle},
\end{equation}
where $\rho(t)^{d-1}$ approximates the probability that all other $d-1$ members of a hyperedge are infected. The corresponding evolution of the mean infection density is
\begin{equation}
\rho(t+1) = (1-\rho(t))\big(1-q(t)\big) + (1-\mu)\rho(t).
\label{eq:rho_t}
\end{equation}

Averaging the hyperedge activity dynamics gives
\begin{equation}
\Theta(t+1) = \Theta(t)\big(1 - \eta \rho(t)^d\big) + \gamma \big(1-\Theta(t)\big),
\end{equation}
with $\rho(t)^d$ approximating the average infection pressure. At equilibrium where $\rho(t+1)=\rho(t)=\rho$, $\Theta(t+1)=\Theta(t)=\Theta^{\star} $, the stationary hyperedge activity is
\begin{equation}
\Theta^{\star} = \frac{\gamma}{\gamma + \eta \rho^d}.
\end{equation}

Substituting $\Theta^{\star}$ into Eq.~\eqref{eq:rho_t}, we set up a self-consistent equation for $\rho$ in the stationary state as
\begin{equation}
\rho = \frac{1-\rho}{\mu} \left[ 1 - \left( 1 - \frac{\gamma \beta \rho^{d-1}}{\gamma + \eta \rho^d} \right)^{\langle k \rangle} \right].
\label{eq:self}
\end{equation}

Using the binomial expansion $(1-\epsilon)^n = 1 - n\epsilon + \tfrac{n(n-1)}{2}\epsilon^2 + O(\epsilon^3)$ to expand near $\rho\to 0$, 
we define the self-consistency function $G(\rho)$ as
\begin{equation}
\label{eq:G}
G(\rho) = \frac{\langle k \rangle}{\mu} \beta \, \rho^{d-1} + o(\rho^{d-1}),
\end{equation}
since $\Theta^{\star}\to 1$ as $\rho\to 0$. 

For $d\ge 3$, the leading-order term in Eq.~\eqref{eq:G}  scales as $\rho^{d-1}$ with $d-1\ge 2$, so the linear term vanishes and $\rho=0$ is stable against infinitesimal perturbations. Non-zero equilibria therefore arise only via a saddle-node bifurcation at finite $\beta$. Thus, for $d \geq 3$, the system exhibits bistability and abrupt first-order transitions at $\beta_n$. In this case, the nonlinear threshold $\beta_n$ is identified by rewriting the self-consistent equation of Eq.~\eqref{eq:self} as a parametric function:
\begin{equation}
\beta_n = \frac{1 - \left( 1 - \frac{\mu \rho}{1-\rho} \right)^{1/\langle k\rangle}}{\Theta^{\star}\rho^{d-1}}.
\end{equation}
where $\Theta^{\star} = \gamma / (\gamma + \eta \rho^d)$ and $\beta_n$ corresponds to the smallest value for which endemic equilibria exist, marking the onset of endemic persistence.
\hfill$\square$

\medskip
\noindent
\textbf{Corollary 1.}  
For pairwise spreading dynamics $(d=2)$, we get the classical epidemic threshold $\beta_c$ with the linear form
\begin{equation}
\beta_c = \frac{\mu}{\Lambda_{\max}(\mathbf{M})},
\end{equation}
where $\Lambda_{\max}(\mathbf{M})$ denotes the largest eigenvalue of the adjacency matrix of the underlying network.

\medskip
\noindent
\textbf{Proof.}
When $d=2$, the self-consistent mapping $G(\rho)$ in Eq.~\eqref{eq:G} expands linearly around $\rho=0$:
\begin{equation}
G(\rho) = \frac{\langle k\rangle}{\mu}\,\beta\, \rho + o(\rho).
\end{equation}
The disease-free equilibrium $\rho=0$ is locally stable if small perturbations decay, which in a one-dimensional map requires
$ G'(0) < 1$. Substituting the linear expansion gives
\begin{equation}
\frac{\langle k\rangle}{\mu}\,\beta < 1,
\end{equation}
leading to the classical threshold
\begin{equation}
\beta_c = \frac{\mu}{\langle k\rangle}.
\label{eq:beta_c_linear}
\end{equation}

Under the homogeneous-degree approximation, we have 
$\Lambda_{\max}(\mathbf{M}) \approx \langle k\rangle$, and thus 
Eq.~\eqref{eq:beta_c_linear} may be equivalently expressed as
\begin{equation}
\beta_c = \frac{\mu}{\Lambda_{\max}(\mathbf{M})}.
\end{equation}
\hfill$\square$

\medskip
\noindent
\textbf{Remark 1.}  
The distinction between linear and nonlinear epidemic thresholds stems from fundamental differences in early-stage dynamics: pairwise interactions ($d=2$) permit linear approximation and spectral characterization, whereas genuine higher-order interactions ($d\geq 3$) eliminate linear instability at the disease-free state due to the $\rho^{d-1}$ scaling. This mathematical foundation explains the emergence of discontinuous transitions and bistability in higher-order contagion processes—phenomena absent in classical pairwise models. Collectively, these results establish that adaptive hyperedge activity fundamentally transforms classical epidemic thresholds and induces nonlinear, bistable regimes in systems with higher-order interactions, providing a theoretical basis for understanding abrupt epidemic outbreaks and hysteresis effects in complex social systems.

\section{Numerical simulations}		
\label{sec:simulations}

In this section, we conduct comprehensive numerical simulations of the adaptive dynamics between higher-order network topology and epidemic spreading, building upon the theoretical framework established above. 
Our experiments systematically examine how the topology-dynamics coupling shapes epidemic outcomes across four complementary dimensions: 1) verification of discontinuous phase transitions; 2) characterization of adaptive hyperedge dynamics; 3) evaluation of immunization strategies; 4) assessment of rewiring mechanisms; and 5) validation on empirical higher-order collaboration networks. 
In the $s$-SIS model the fully susceptible state is absorbing, therefore finite-size simulations tend to prematurely freeze in the disease-free trap.
To circumvent this and obtain reliable stationary averages, we adopt the quasistationary (QS) method~\cite{ferreira2011quasistationary}.
Following the original recommendation, and confirmed by our own preliminary runs, the system is kept active by storing a history of 50 previously visited configurations; upon absorption, the current state is replaced with a randomly drawn historical one with probability $0.2$.

Unless otherwise specified, we set the network size with $n = 1500$ nodes with hyperedge size $d = 3$. The epidemic dynamics are parameterized with infection rate $\beta = 0.2$, average hyperdegree $\langle k \rangle = 6$,  recovery rate $\mu = 0.1$, 
while the adaptive hyperedge dynamics are governed by the sensitivity parameter $\eta = 0.6$, which captures how promptly groups reduce activity upon perceiving infection risk, and the spontaneous recovery rate $\gamma = 0.02$.
Simulations are initialized with an initial infection density $I(0) = 0.3$. The offset parameter for degree-preferential rewiring is set to $\alpha = 1.0$. All results obtained through theoretical or MC simulations represent averages over 200 independent simulations.
 
\subsection{Verification of discontinuous phase transitions}	

\begin{figure}[htbp]
  \centering
  \hspace*{-0.5cm}
  \begin{overpic}[width=1.05\linewidth]{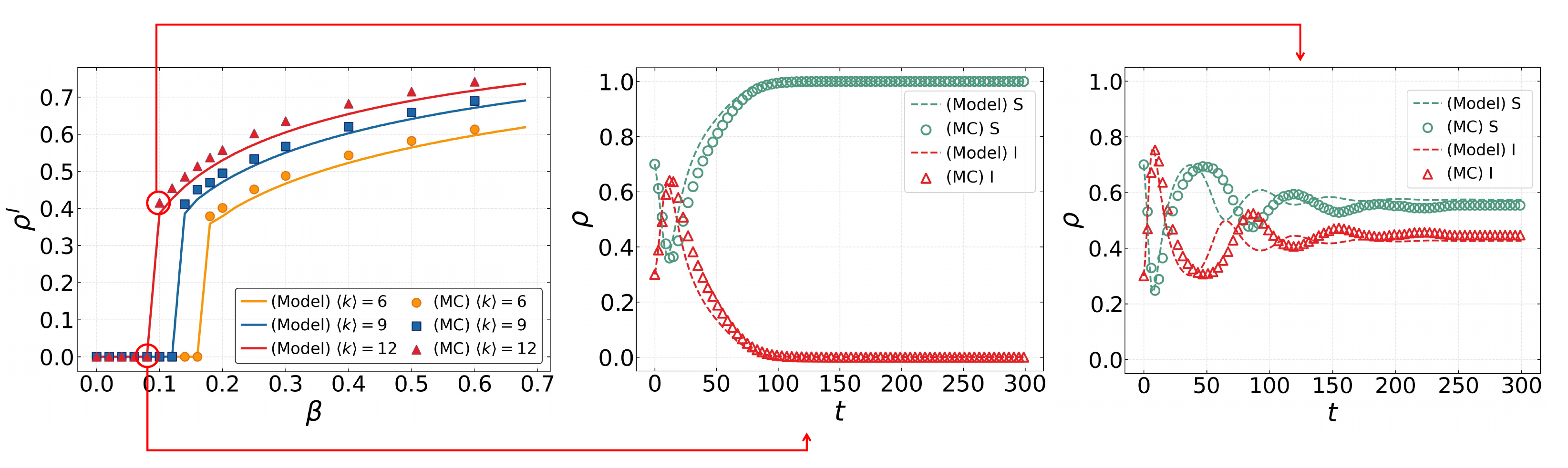}
    \put(18,-1.2){\textbf{(a)}}
    \put(52,-1.2){\textbf{(b)}}
    \put(84,-1.2){\textbf{(c)}}
  \end{overpic}

   \vspace{0.3cm} 
   \caption{
        \textbf{Comparison between theoretical predictions from MMCA equations and MC simulations on adaptive hypergraphs.} 
        Panel (a) shows the infection density $\rho^I$ in the stationary state as a function of infection rate $\beta$ under different average hyperdegrees $\langle k \rangle =6, 9, 12$. 
        Panels (b) and (c) display the temporal evolution of susceptible (S) and infected (I) fractions $\rho$ for $\langle k \rangle = 12$ at $\beta = 0.08$ and $\beta = 0.10$, respectively. 
        Solid curves denote MMCA theoretical predictions, and symbols represent MC simulation averages over $200$ independent simulations.
    }
    
   \label{fig_1}
  \end{figure}


\begin{figure*}[ht!]
    \centering
    \begin{subfigure}[b]{0.345\textwidth}
        \centering
        \includegraphics[width=\textwidth]{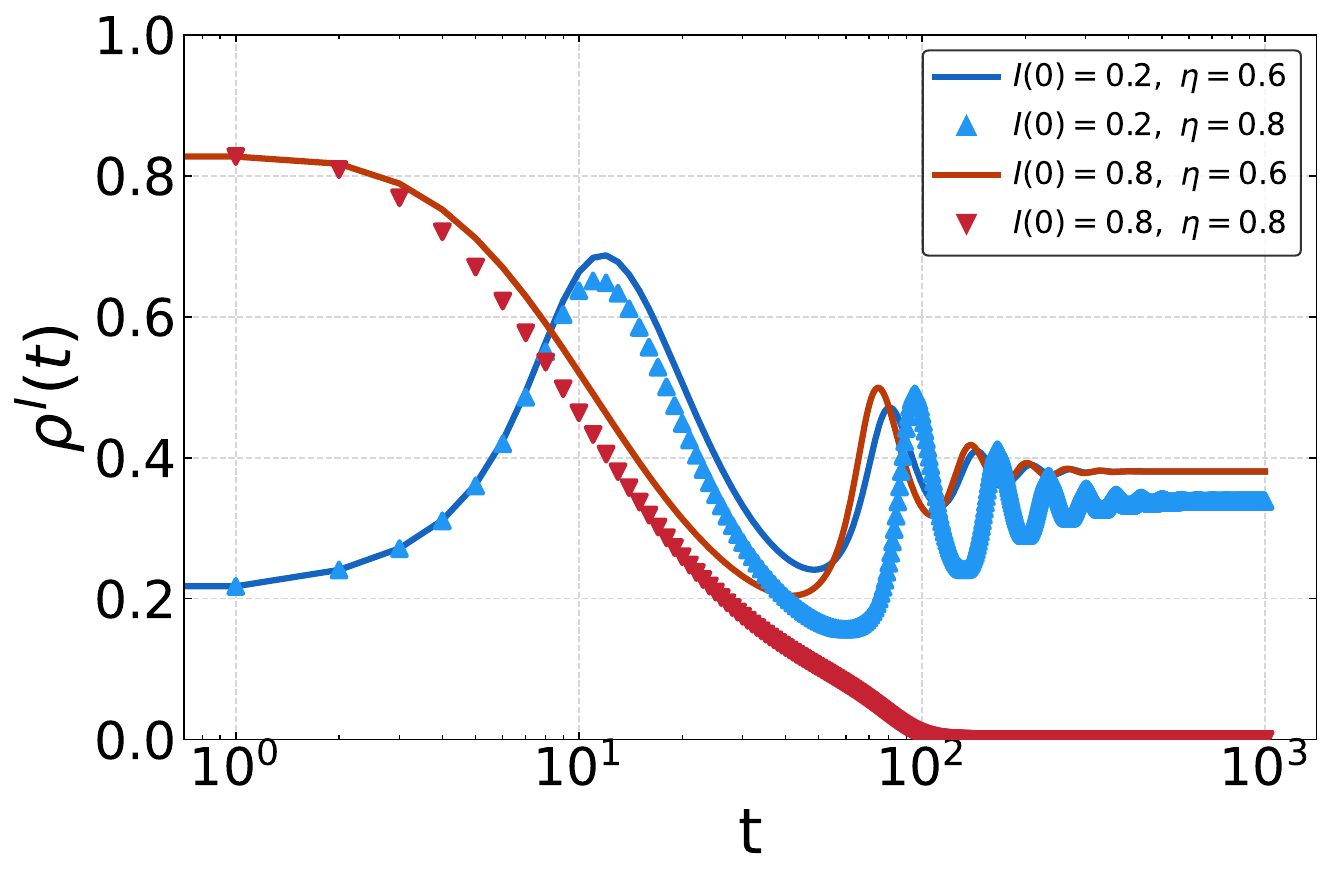}
        \caption{}
        \label{fig:char_adaptive_a}
    \end{subfigure}
    \hfill
    \begin{subfigure}[b]{0.315\textwidth}
        \centering
        \includegraphics[width=\textwidth]{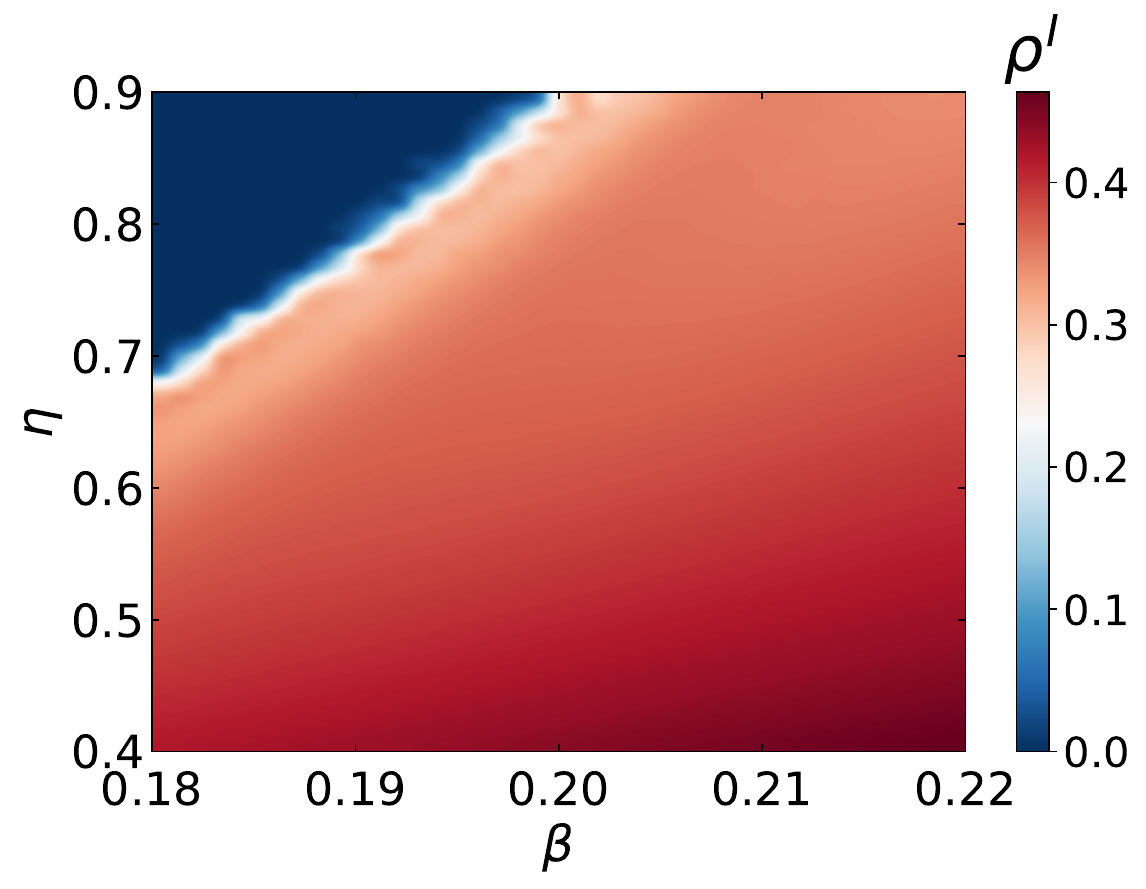}
        \caption{}
        \label{fig:char_adaptive_b}
    \end{subfigure}
    \hfill
    \begin{subfigure}[b]{0.315\textwidth}
        \centering
        \includegraphics[width=\textwidth]{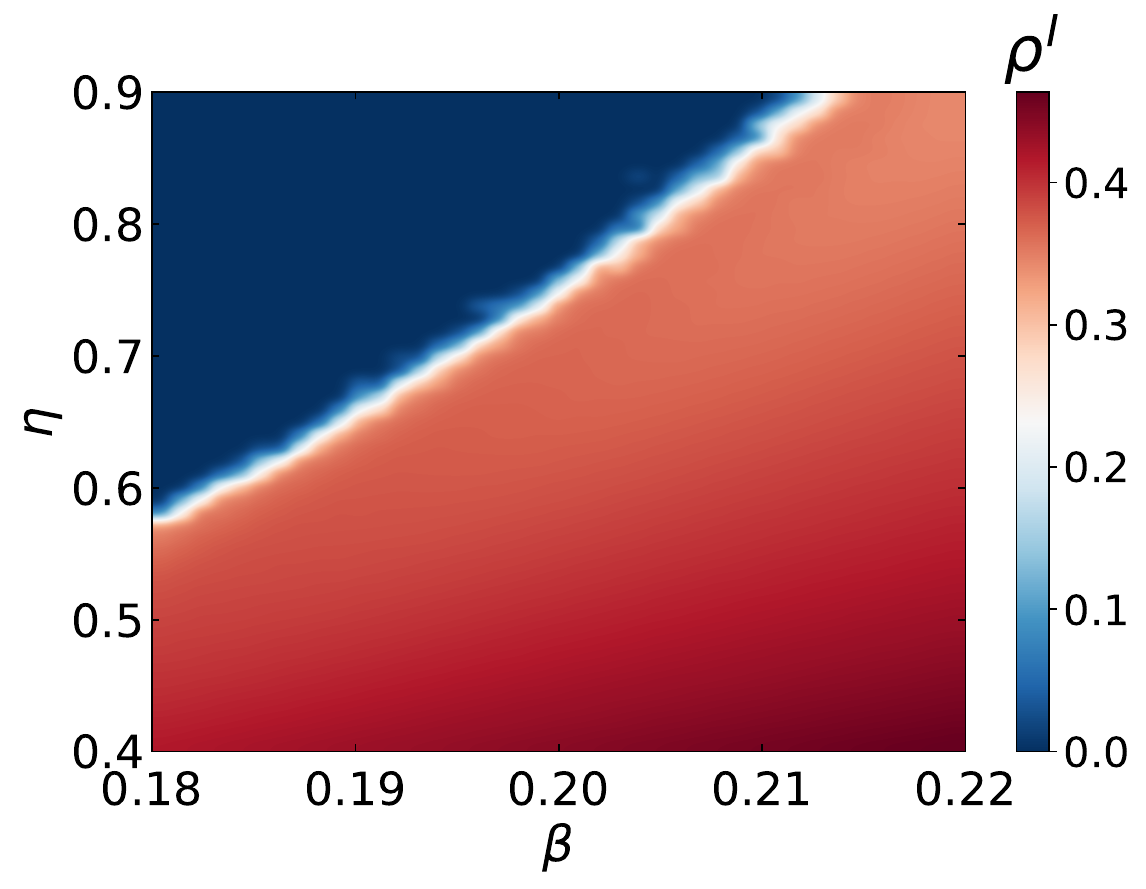}
        \caption{}
        \label{fig:char_adaptive_c}
    \end{subfigure}
    
    \begin{subfigure}[b]{0.345\textwidth}
        \centering
        \includegraphics[width=\textwidth]{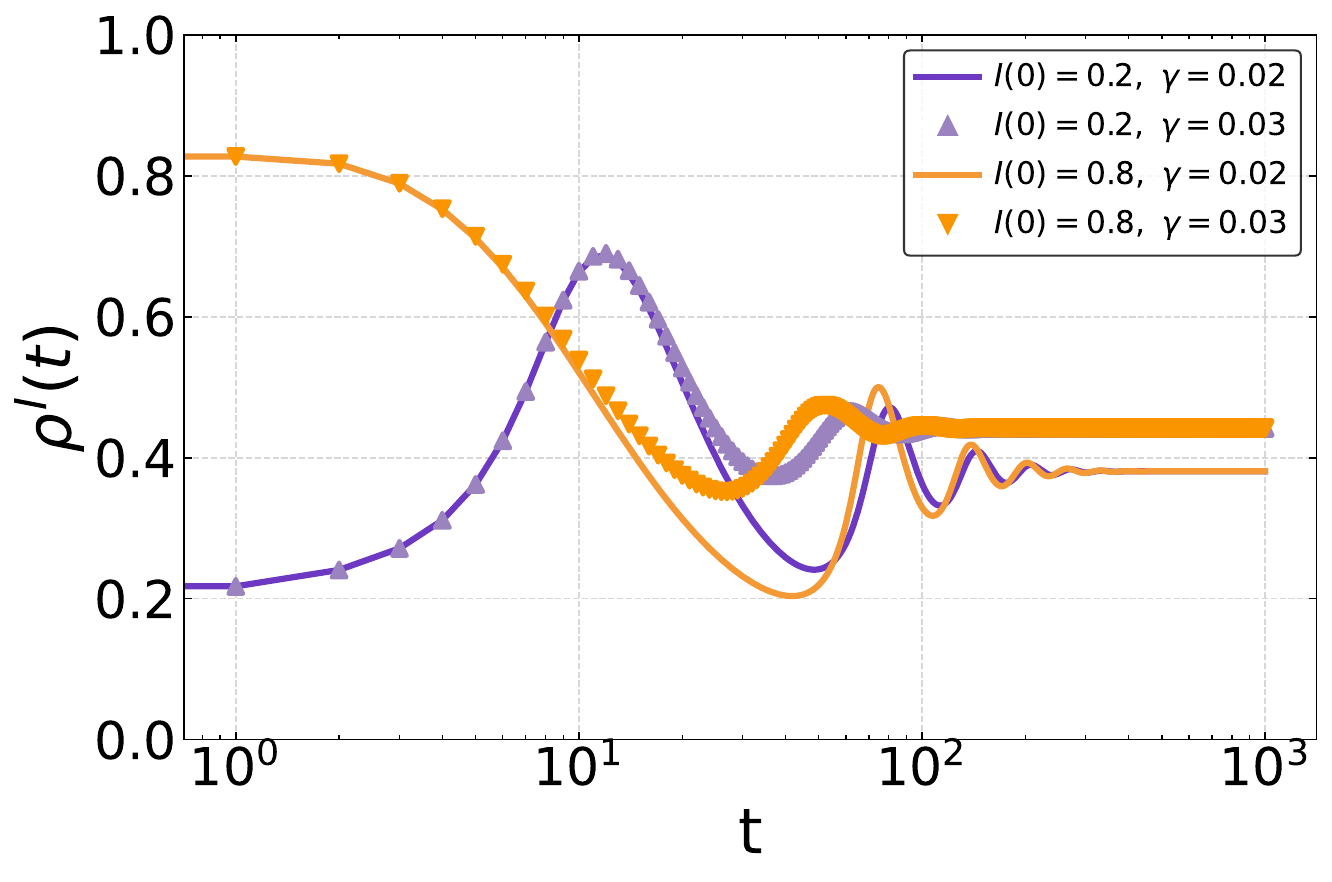}
        \caption{}
        \label{fig:char_adaptive_d}
    \end{subfigure}
    \hfill
    \begin{subfigure}[b]{0.315\textwidth}
        \centering
        \includegraphics[width=\textwidth]{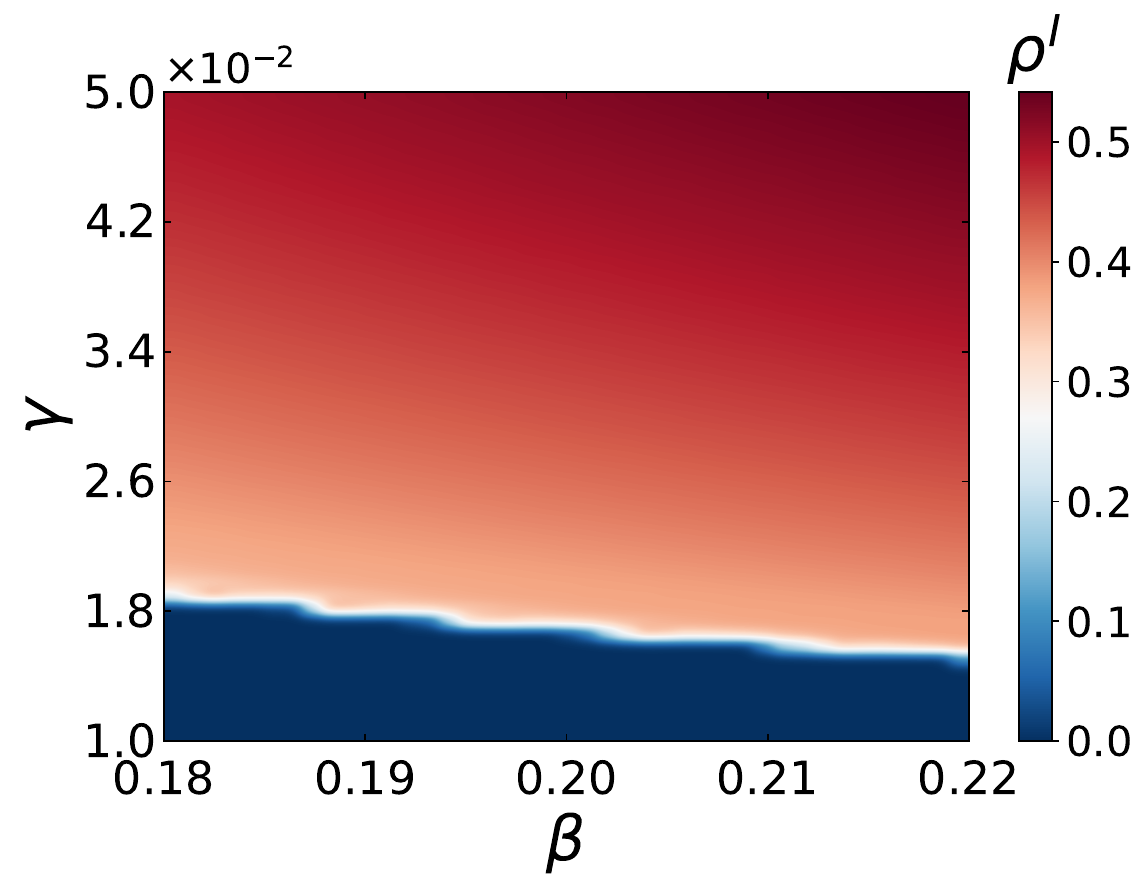}
        \caption{}
        \label{fig:char_adaptive_e}
    \end{subfigure}
    \hfill
    \begin{subfigure}[b]{0.315\textwidth}
        \centering
        \includegraphics[width=\textwidth]{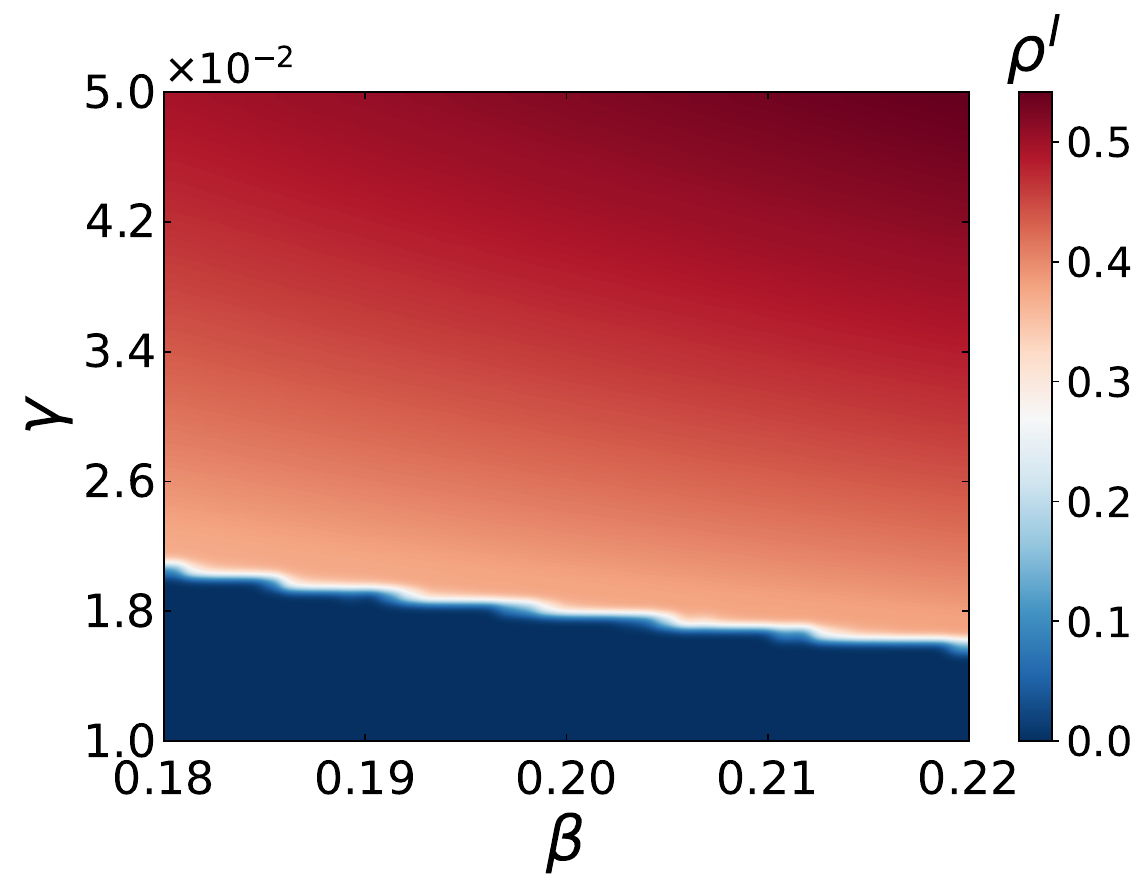}
        \caption{}
        \label{fig:char_adaptive_f}
    \end{subfigure}
    
    \caption{
        \textbf{The effects of initial conditions and adaptive parameters on outbreak dynamics and stationary-state prevalence.} 
        Time evolution of the infection density $\rho^I(t)$ over $10^3$ time steps under different initial conditions $I(0)$, shown for (a) varying sensitivity $\eta$ and (d) varying spontaneous recovery rate $\gamma$. 
        Panels (b) and (c) give the stationary-state infection density $\rho^I$ as a function of $\beta$ and $\eta$ for $I(0)=0.2$ and $I(0)=0.8$, while panels (e) and (f) show the corresponding stationary prevalence as a function of $\beta$ and $\gamma$ for the same two initial conditions. 
        The color bars indicate the final infection density, with dark blue marking disease-free states.}
    \label{fig_2}
\end{figure*}

To validate the effectiveness of the proposed MMCA formulation in reproducing epidemic dynamics on adaptive hypergraphs, we conduct a comparative analysis between theoretical predictions and MC simulations. Specifically, we investigate the stationary-state infection density $\rho$ as a function of the infection rate $\beta$ under three different average hyperdegrees $\langle k \rangle = 6, 9, 12$. 
The results are displayed in Fig.~\ref{fig_1}. 
The experimental setup enables us to evaluate the theoretical epidemic threshold and verify the existence of a first-order phase transition predicted by the MMCA analysis.

As illustrated in Fig.~\ref{fig_1}, the theoretical curves and simulation results exhibit strong agreement across all parameter regimes, confirming the validity of the MMCA in capturing macroscopic epidemic dynamics on hypergraphs. 
In Fig.~\ref{fig_1}(a), both theoretical and numerical results demonstrate a clear discontinuous transition: when $\beta$ exceeds a critical value $\beta_n$, the system abruptly shifts from a disease-free to an endemic state. 
This behavior is fully consistent with our theoretical analysis in 
Theorem~\ref{theory1}, which establishes that higher-order interactions induce a 
first-order (discontinuous) phase transition and a genuinely nonlinear 
epidemic threshold arising from a saddle-node bifurcation.
Moreover, the stationary infection density increases with the average hyperdegree $\langle k \rangle$, indicating that enhanced group connectivity facilitates persistent infection through multiple higher-order transmission pathways.

To further elucidate the temporal evolution near the critical region, Figs.~\ref{fig_1}(b) and (c) represent the time evolution of each compartmental fraction for $\langle k \rangle=12$ at $\beta=0.08$ and $\beta=0.10$, respectively. For $\beta=0.08 < \beta_n$, the infection gradually vanishes, and the system converges to a disease-free equilibrium after mild transient oscillations. In contrast, when $\beta=0.10 > \beta_n$, the infection persists and stabilizes at a finite endemic level after pronounced damped oscillations.
These oscillatory patterns originate from the intrinsic feedback loop between epidemic spreading and adaptive topology.  
An increase in infection density enhances the effective infection pressure $\Phi_e(t)$ across active hyperedges, which in turn suppresses their activity levels $\Theta_e(t)$ through the adaptive rule in Eq.~\eqref{eq:Theta_update}.  
The ensuing reduction in hyperedge activity lowers the overall transmission capacity, leading to a subsequent decline in infection prevalence.  
As the epidemic subsides, hyperedges gradually regain activity via spontaneous recovery, thereby reopening transmission channels and initiating the next fluctuation cycle.  
The interplay between node-level infection and group-level adaptation produces self-regulated oscillations around the endemic equilibrium and gives rise to the hysteresis phenomena observed in adaptive higher-order contagion systems.  
Such behavior highlights that epidemic persistence in hypergraphs is not solely determined by transmission rates but also by the dynamical coupling between contagion and structural adaptation, which jointly shape the resilience of networks and transition pathways.

\begin{figure}[ht!]
    \centering

    \begin{subfigure}[b]{0.48\textwidth}
        \centering
        \includegraphics[width=\textwidth]{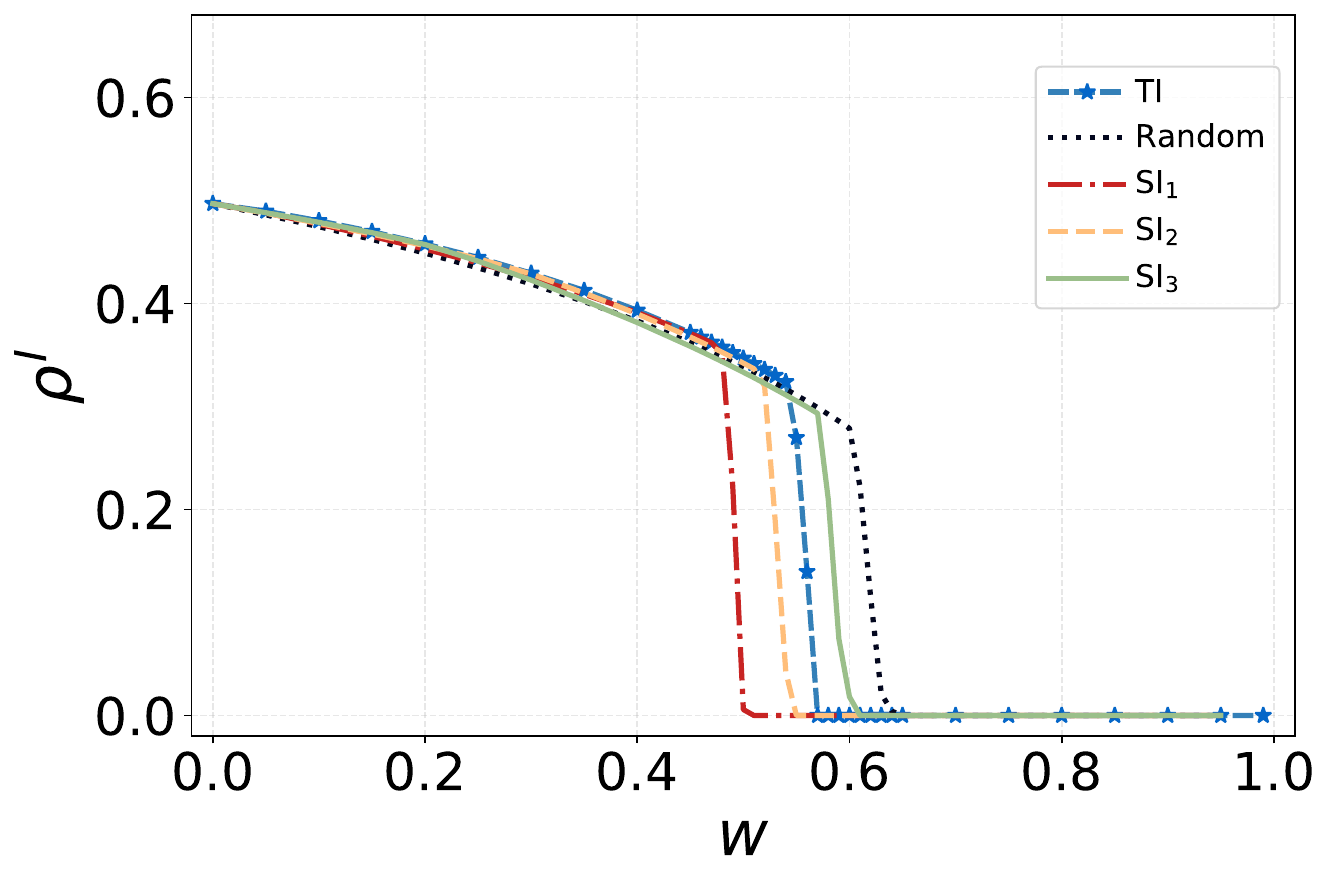}
        \caption{$\langle k \rangle = 6,\, \gamma = 0.02$}
        \label{fig:immunization_a}
    \end{subfigure}
    \hfill
    \begin{subfigure}[b]{0.48\textwidth}
        \centering
        \includegraphics[width=\textwidth]{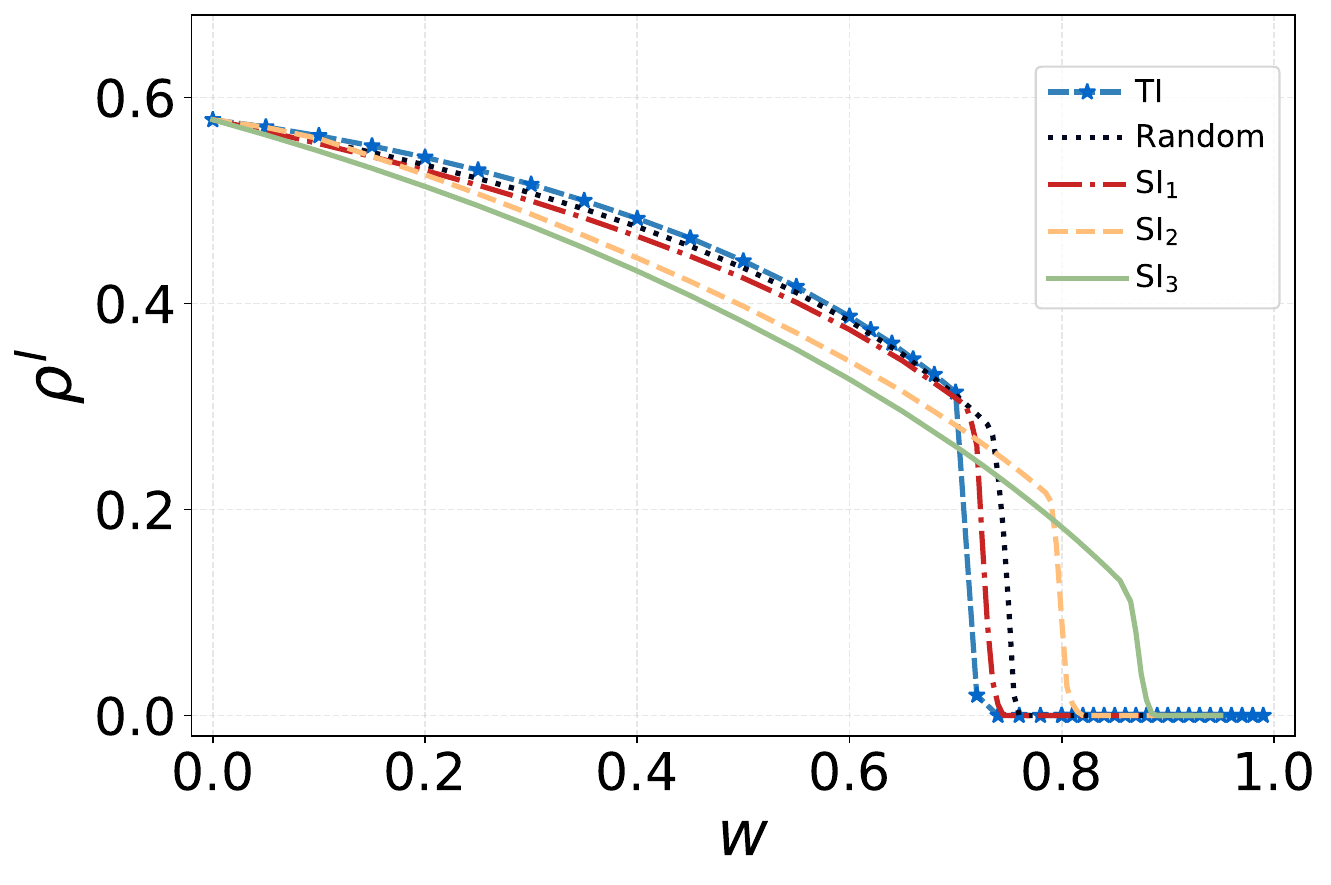}
        \caption{$\langle k \rangle = 9,\, \gamma = 0.02$}
        \label{fig:immunization_b}
    \end{subfigure}

    \vspace{0.2em}
    \begin{subfigure}[b]{0.48\textwidth}
        \centering
        \includegraphics[width=\textwidth]{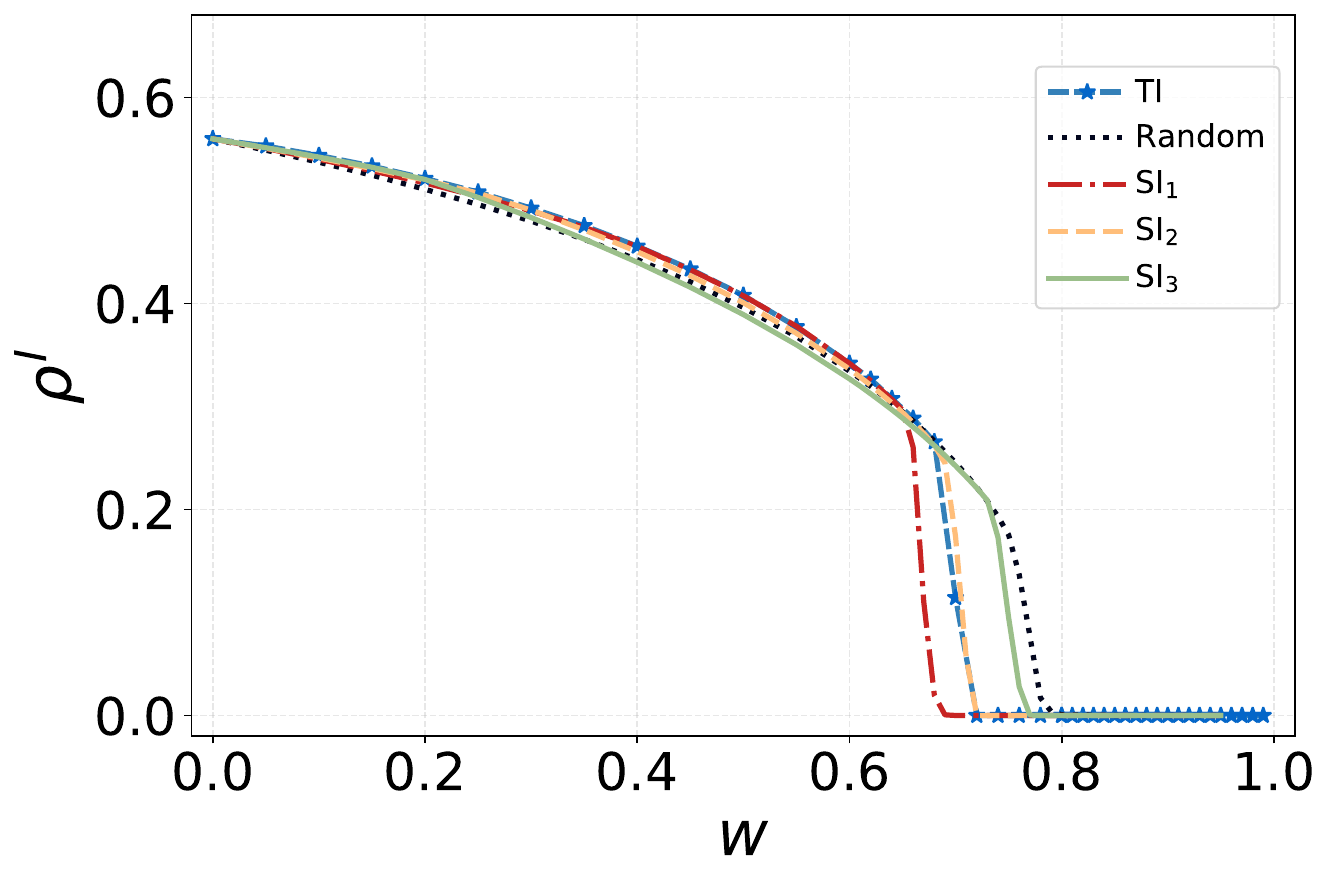}
        \caption{$\langle k \rangle = 6,\, \gamma = 0.03$}
        \label{fig:immunization_c}
    \end{subfigure}
    \hfill
    \begin{subfigure}[b]{0.48\textwidth}
        \centering
        \includegraphics[width=\textwidth]{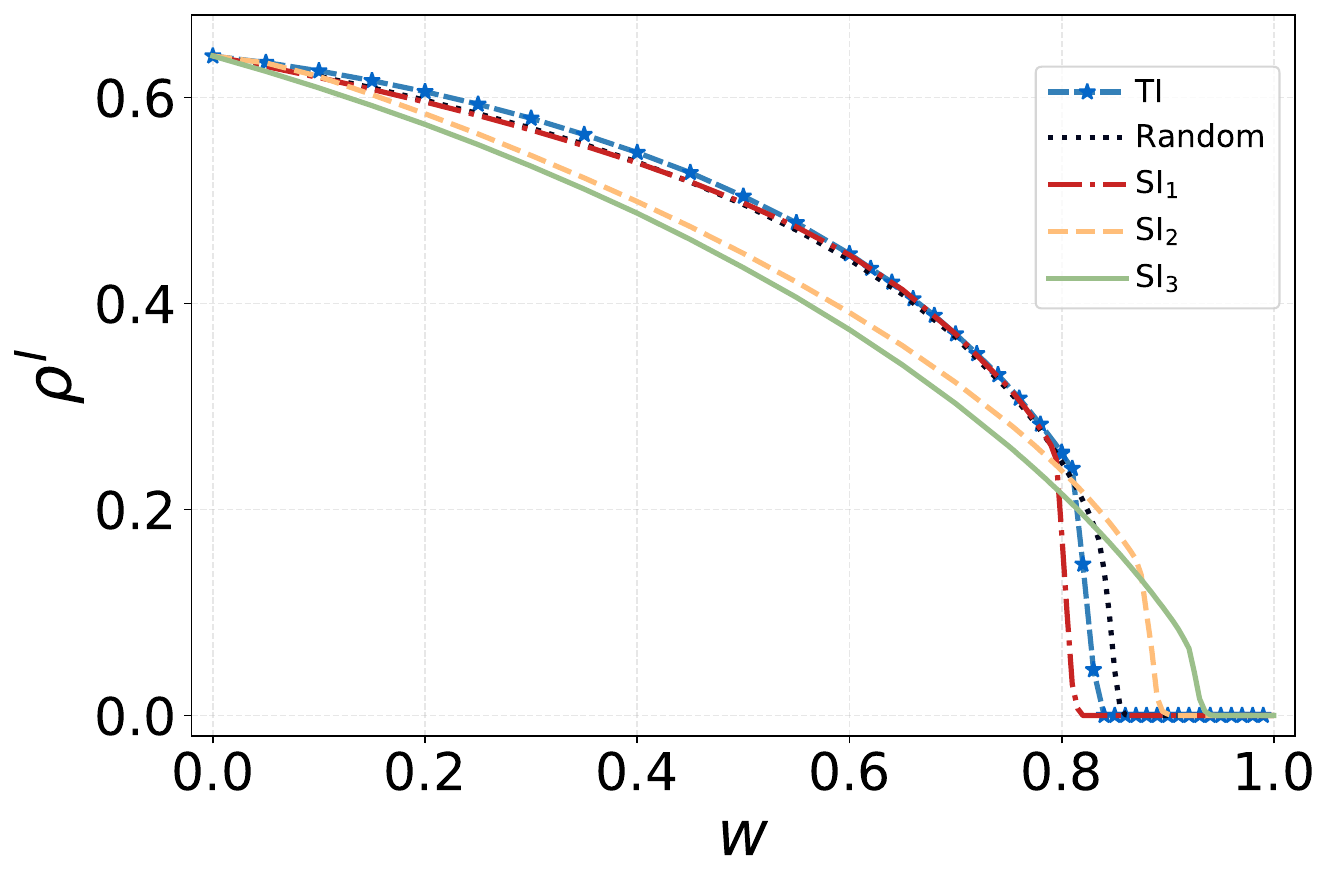}
        \caption{$\langle k \rangle = 9,\, \gamma = 0.03$}
        \label{fig:immunization_d}
    \end{subfigure}
    
    \caption{
      \textbf{Comparative effectiveness of immunization strategies.} 
      Performance comparison among Targeted immunization (TI), Random edge immunization (Random), and Spontaneous isolation (SI) strategies under varying activity thresholds $\theta_{\mathrm{min}}=0.2, 0.4, 0.6$ for SI\textsubscript{1}, SI\textsubscript{2}, and SI\textsubscript{3}, respectively. 
      Panels (a)–(d) correspond to four parameter settings: (a) $\langle k \rangle = 6,\, \gamma = 0.02$; (b) $\langle k \rangle = 9,\, \gamma = 0.02$; (c) $\langle k \rangle = 6,\, \gamma = 0.03$; and (d) $\langle k \rangle = 9,\, \gamma = 0.03$. 
      Each panel depicts the stationary infection density $\rho^I$ as a function of the immunization rate $w$, defined as the fraction of hyperedges removed or deactivated through immunizations. 
      The infection rate is fixed at $\beta = 0.35$ across all simulations.
    }
    \label{fig_3}
\end{figure}

\begin{figure}[ht!]
    \centering
    \begin{subfigure}[b]{0.48\textwidth}
        \centering
        \includegraphics[width=\textwidth]{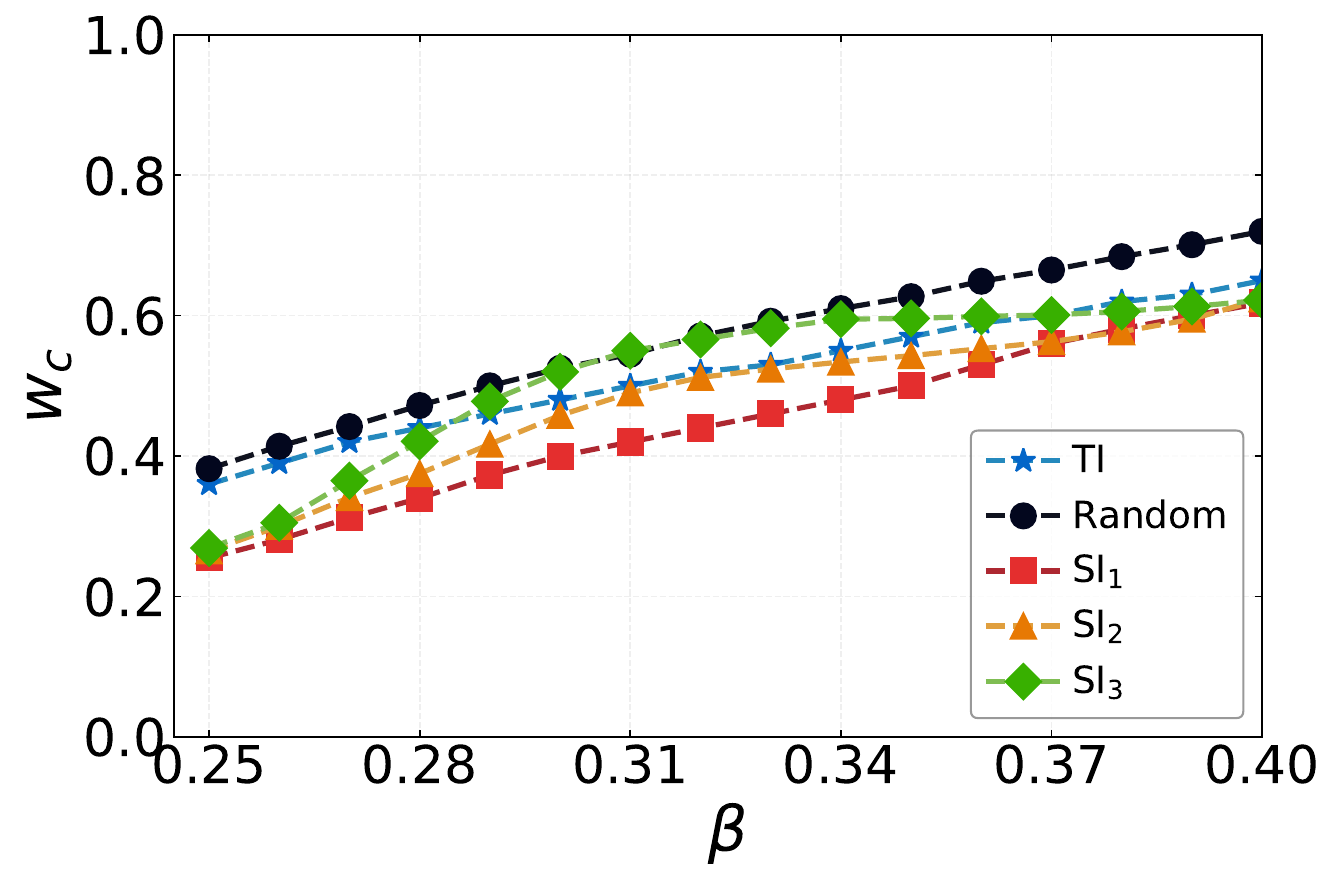}
        \caption{$I(0) = 0.3$}
        \label{fig:HIT_a}
    \end{subfigure}
    \hfill
    \begin{subfigure}[b]{0.48\textwidth}
        \centering
        \includegraphics[width=\textwidth]{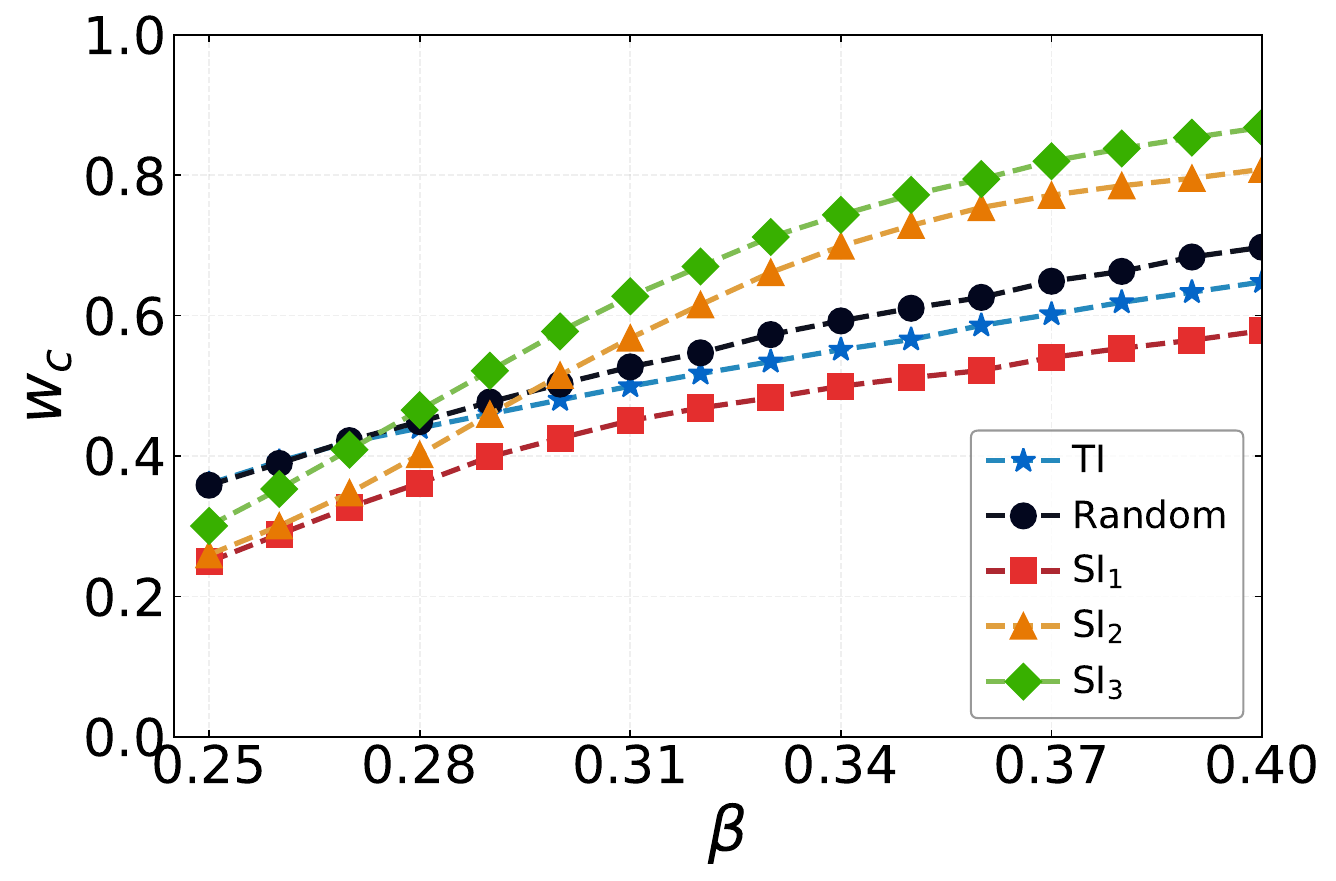}
        \caption{$I(0) = 0.5$}
        \label{fig:HIT_b}
    \end{subfigure}
    
    \caption{
        \textbf{Herd immunity threshold $w_c$ under different immunization strategies.} 
        The critical immunization fraction $w_c$, defined as the minimum proportion of immunized hyperedges required to eliminate epidemics, is shown as a function of the infection rate $\beta$ for two initial infection levels: (a) $I(0) = 0.3$ and (b) $I(0) = 0.5$. 
        Results are compared across Targeted immunization (TI), three variants of Spontaneous isolation (SI), and Random edge immunization (Random), represented by solid curves with distinct colors and markers. 
        The average hyperdegree is set to $\langle k \rangle = 6$, and all other parameters follow the baseline configuration.
    }
    \label{fig_4}
\end{figure}

\subsection{Characterization of adaptive hyperedge dynamics}		
 
To systematically examine how adaptive hyperedge behavior shapes epidemic outcomes in hypergraphs, we analyze the joint effects of initial conditions and adaptive parameters on outbreak dynamics and stationary-state prevalence through time-series comparisons and heatmaps in Fig.~\ref{fig_2}.

Figs.~\ref{fig_2}(a) and (d) depict the temporal evolution of infection density $\rho^I(t)$ under varying initial conditions and adaptive parameters. Interestingly, as shown in panel (a), for high sensitivity $\eta = 0.8$, a large initial infection density $I(0) = 0.8$ leads to disease extinction, whereas a moderate initial prevalence $I(0) = 0.2$ results in endemic persistence. This counter-intuitive behavior manifests as bistability and hysteresis phenomena, where the system exhibits two stable stationary states, disease-free and endemic, for identical parameter values, with the final outcome determined solely by initial conditions. The bistable regime arises from the adaptive feedback between node states and hyperedge activity. Specifically, high initial infection density induces strong infection pressure $\Phi_e(t)$ on hyperedges, which under high sensitivity rapidly suppresses hyperedge activity $\Theta_e(t)$, ultimately extinguishing the outbreak through self-induced topological isolation.
Moreover, 
transient oscillations arise from delayed feedback between infection growth and structural adaptation, producing damped oscillations that highlight the rich phenomenology of adaptive dynamics in higher-order networks.

Furthermore,
the related parameter analysis in Fig.~\ref{fig_2}(b)--(f) quantifies these adaptive mechanisms.
Heatmaps (b) and (c) reveal that increasing the sensitivity parameter $\eta$ and initial infection density $I(0)$ systematically suppresses endemic prevalence, effectively expanding the disease-free region in the heatmaps. Conversely, heatmaps (e) and (f) demonstrate that a higher spontaneous recovery rate $\gamma$ leads to increased epidemic prevalence by accelerating hyperedge reactivation. These opposing effects reflect the dual role of adaptive dynamics, as strong infection-induced suppression with high sensitivity $\eta$ provides natural containment, while rapid structural recovery with $\gamma$ maintains transmission potential.  

\subsection{Evaluation of immunization strategies}		

Next, to evaluate the performance of hyperedge-level immunization strategies in adaptive hypergraphs, we compare the stationary-state infection density $\rho^I$ as a function of the immunization rate $w$ across structural settings $\langle k \rangle $ and adaptive parameters $\gamma$. The results are displayed in Fig.~\ref{fig_3}. As shown in Fig.~\ref{fig_3}(a)–(b), all strategies perform similarly to random immunization at small $w$, but once $w$ exceeds a critical threshold, $\rho^I$ drops sharply, reaching herd immunity at a lower $w_c$ than the random baseline.

In dense networks with faster recovery shown in Figs.~\ref{fig_3}(b) and (d), SI\textsubscript{2} and SI\textsubscript{3} outperform other strategies under limited budgets ($w < 0.6$). This effect occurs because these strategies preferentially remove highly active hyperedges, which have relatively low instantaneous infection pressure but play an important role as potential transmission pathways in highly connected systems.  
Specifically, the higher threshold of SI\textsubscript{3} allows it to quickly block epidemiologically active hyperedges at low budgets, explaining its early advantage. 
As the budget grows, however, the expansion of immunization to less active hyperedges is delayed by SI\textsubscript{3}'s descending activity-based ranking, which postpones intervention on those hyperedges with the highest infection pressure but severely suppressed activity. 
These residual hyperedges sustain endemic circulation and prevent full elimination. 
In contrast, SI\textsubscript{1} targets such high-risk hyperedges from the start, achieving lower stationary prevalence at sufficiently large budgets.
When $w$ increases further, TI and SI\textsubscript{1} become more effective since they directly target hyperedges with high infection pressure or lower activity levels, respectively. 
SI\textsubscript{2} and SI\textsubscript{3} continue to prioritize hyperedges with high activity but low infection pressure, leading to diminishing returns in the later stage.
These results reveal a key trade-off in adaptive hypergraphs: proactive isolation of potential transmission hubs enhances resilience under resource constraints, whereas targeted interventions guided by infection pressure are crucial for complete epidemic elimination.

Then, we examine how different immunization strategies influence the herd immunity threshold (HIT) in adaptive hypergraphs.
This set of simulations quantifies the way in which transmission intensity, initial epidemic prevalence, and behavioral thresholds jointly determine the minimal intervention effort required to halt an outbreak in higher-order adaptive networks.
As shown in Fig.~\ref{fig_4}, the HIT \(w_c\) increases monotonically with the infection rate \(\beta\) for every strategy, confirming that stronger transmission demands a larger fraction of immunized hyperedges to prevent endemic persistence.
More importantly, the performance gap separating the strategic interventions (TI and SI\textsubscript{1}) from random immunization widens substantially at higher transmission rates, for instance when \(\beta\) exceeds \(0.31\).
This divergence underscores the rapidly diminishing returns of non-strategic approaches once the epidemic becomes more difficult to contain.

Among all methods examined, SI\textsubscript{1} consistently attains the lowest HIT.
Its effectiveness stems from its ability to identify and disable hyperedges whose activity has been severely suppressed by high local infection pressure, precisely those groups that act as persistent reservoirs of infection.
In Fig.~\ref{fig_4}(b), TI delivers a comparable performance, especially at moderate infection rates, while SI\textsubscript{2} and SI\textsubscript{3} require substantially higher thresholds.
The weaker performance of these latter strategies reflects their slower responsiveness to early signs of activity suppression, which delays the removal of genuinely high-risk hyperedges.
The comparison between Figs.~\ref{fig_4}(a) and (b) further reveals that raising the initial infection level from \(I(0)=0.3\) to \(I(0)=0.5\) elevates the HIT across all strategies, yet the relative advantage of SI\textsubscript{1} is preserved.
This upward shift occurs because a higher initial prevalence triggers stronger endogenous reductions in hyperedge activity, allowing spontaneous isolation to work in concert with the built-in adaptive feedback.
The resulting cooperative coupling amplifies the impact of isolating high-risk hyperedges, delivering improved containment efficiency even when the system starts from a considerably more adverse epidemiological state. This finding reveals that herd immunity in adaptive hypergraphs is not a static network property but an emergent outcome of interacting epidemic, behavioral, and intervention dynamics, and that effective control therefore demands strategies which couple transmission dynamics with adaptive structural interventions.

\begin{figure}[ht!]
    \centering
    \hspace*{1.9em}
    \begin{subfigure}[b]{0.43\textwidth}
        \centering
        \includegraphics[width=\textwidth]{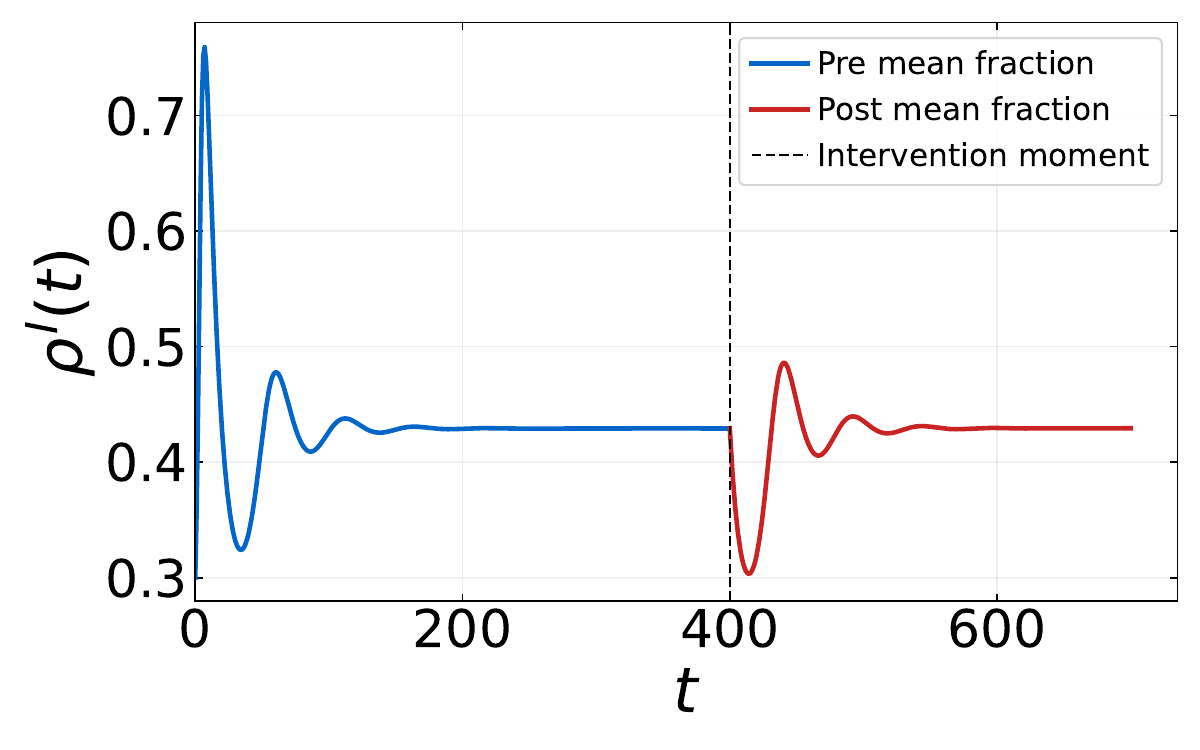}
        \caption{}
        \label{}
    \end{subfigure}
    \hfill
    \begin{subfigure}[b]{0.49\textwidth}
        \centering
        \includegraphics[width=\textwidth]{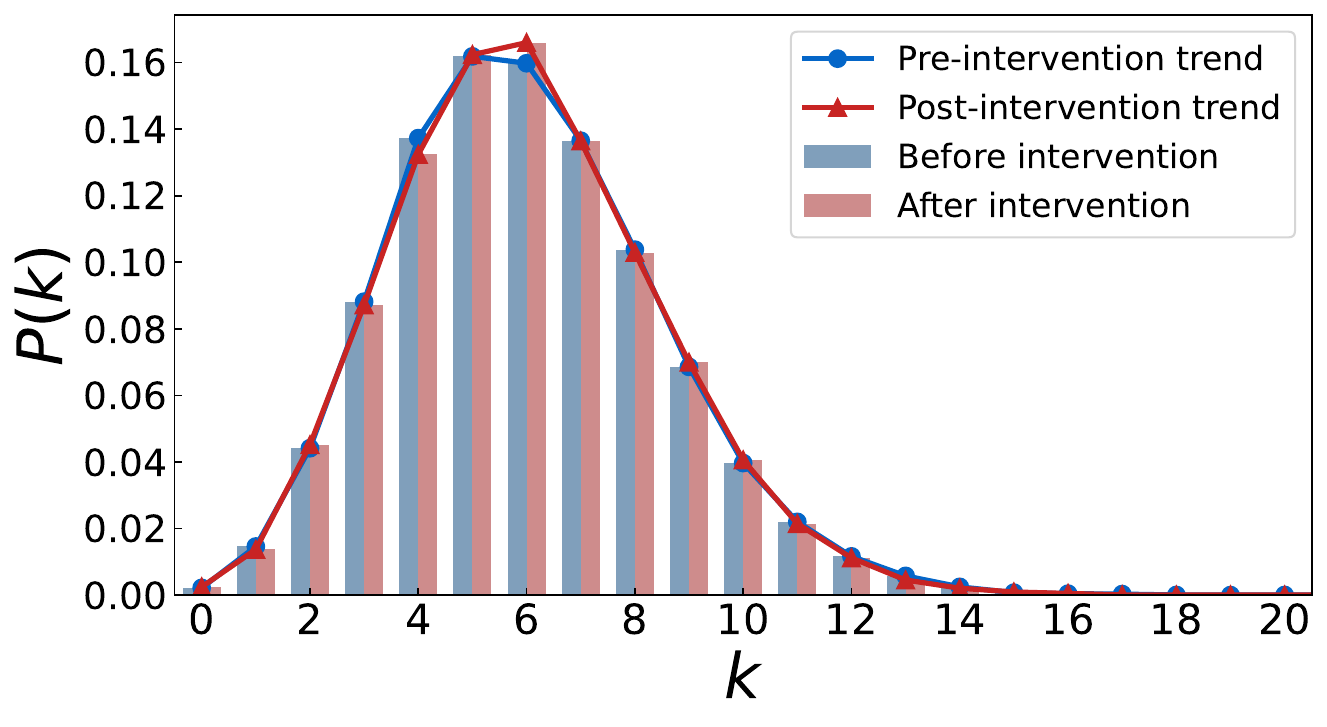}
        \caption{}
        \label{}
    \end{subfigure}
    
    \hspace*{1.9em}
    \begin{subfigure}[b]{0.43\textwidth}
        \centering
        \includegraphics[width=\textwidth]{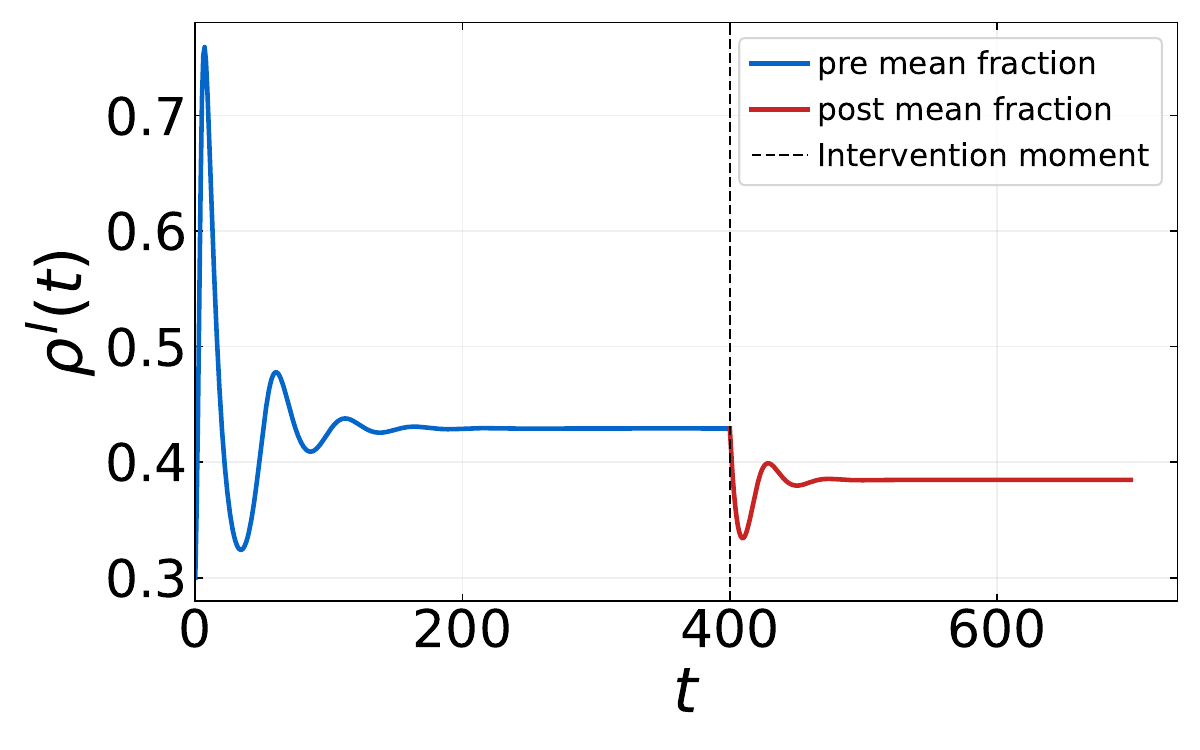}
        \caption{}
        \label{}
    \end{subfigure}
    \hfill
    \begin{subfigure}[b]{0.49\textwidth}
        \centering
        \includegraphics[width=\textwidth]{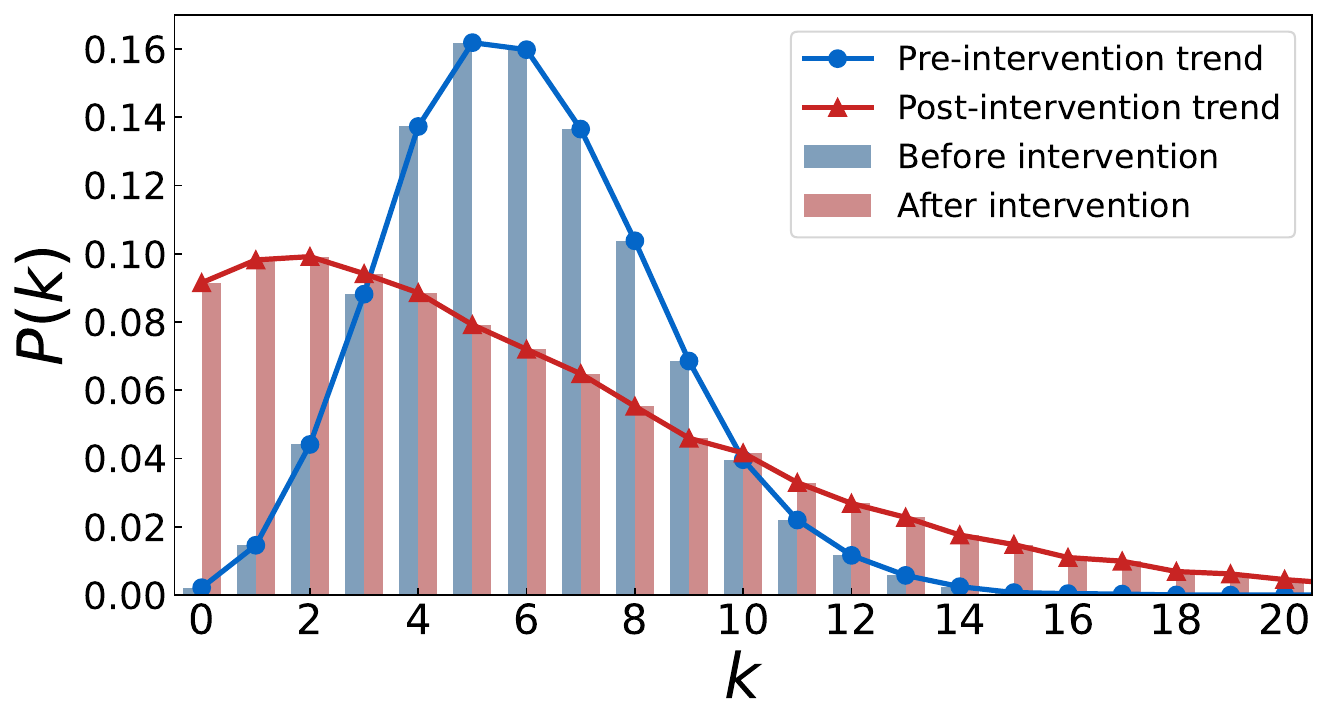}
        \caption{}
        \label{}
    \end{subfigure}
    
    \caption{
       \textbf{Impact of rewiring mechanisms on epidemic dynamics and network topology.} 
       Panels (a) and (c) show the time evolution of infection density under random rewiring and degree-preferential rewiring, respectively, following targeted immunization applied at the intervention moment.
       Panels (b) and (d) compare the degree distributions $P(k)$ before and after random rewiring and degree-preferential rewiring. At time \(t = 400\), the system is perturbed from its quasi-stationary state by removing a fraction \(w = 0.8\) of hyperedges through the TI strategy, after which the two rewiring mechanisms are applied independently in panels (a) and (c).
       Blue curves represent infection dynamics before intervention, and red curves represent dynamics after intervention. The infection rate is fixed at \(\beta = 0.25\), the initial infection density is set to \(I(0) = 0.3\), and all other parameters are kept the same across simulations.
      }
    \label{fig_5}
\end{figure}

\subsection{Assessment of rewiring mechanisms}	
	
To investigate whether structurally guided resumption of group interactions can sustain the control achieved by prior immunization, we apply two classic rewiring mechanisms, random rewiring and degree-preferential rewiring, after the system has reached a quasi-stationary state under targeted immunization.
These rewiring rules model the reopening of social contacts and the re‑establishment of group activities in a post‑outbreak recovery phase, allowing us to ask whether an unstructured return to normal connectivity leads to epidemic resurgence and whether a structured recovery can keep prevalence at a lower level.
The assessment is based on both the temporal evolution of infection density and the changes in degree distribution, and the results are presented in Fig.~\ref{fig_5}.

As illustrated in Fig.~\ref{fig_5}, the two rewiring mechanisms produce qualitatively different epidemic trajectories.
Under random rewiring, shown in panels (a) and (b), the intervention induces only a temporary reduction in infection density.
After this brief suppression, the epidemic rapidly returns to its pre‑intervention level.
Random reconnection restores overall connectivity without structurally blocking transmission pathways, so alternative spreading routes re‑emerge and hyperedge activity gradually recovers, allowing the infection to re‑establish itself.

In contrast, degree‑preferential rewiring, shown in panels (c) and (d), yields a lasting reduction in infection density.
The system remains in a controlled regime well below both the baseline level and the level observed under random rewiring.
The degree distribution provides a structural explanation for this persistence.
Random rewiring keeps the network close to a Poisson‑like pattern, whereas degree‑preferential rewiring creates a heavy‑tailed distribution that reshapes the transmission landscape.
The heavy-tailed degree distribution stems directly from the intrinsic bias of the preferential rewiring kernel, which amplifies existing degree heterogeneity regardless of the infection state.
This heterogeneous structure generates high infection pressure on the few most connected nodes, which in turn triggers stronger adaptive suppression of their incident hyperedges and drives the system toward lower endemic prevalence.
Hence, even in the recovery phase following an outbreak, strategic organisation of contacts and activities remains essential for preventing epidemic resurgence.

\subsection{Validation on empirical higher-order collaboration networks}	
\begin{figure*}[ht!]
    \centering
    \begin{subfigure}[b]{0.325\textwidth}
        \centering
        \includegraphics[width=\textwidth]{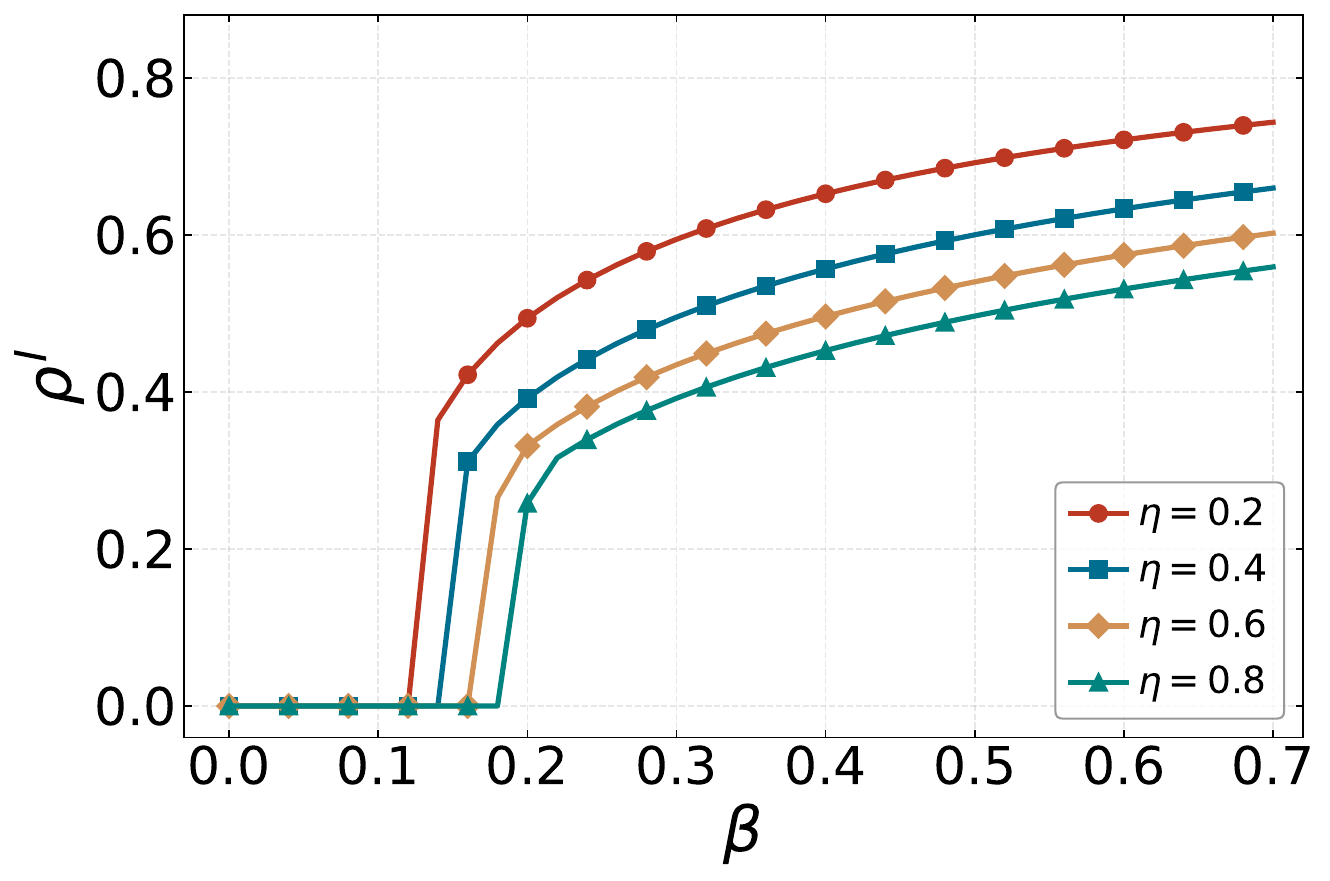}
        \caption{}
        \label{fig:real_a}
    \end{subfigure}
    \hfill
    \begin{subfigure}[b]{0.325\textwidth}
        \centering
        \includegraphics[width=\textwidth]{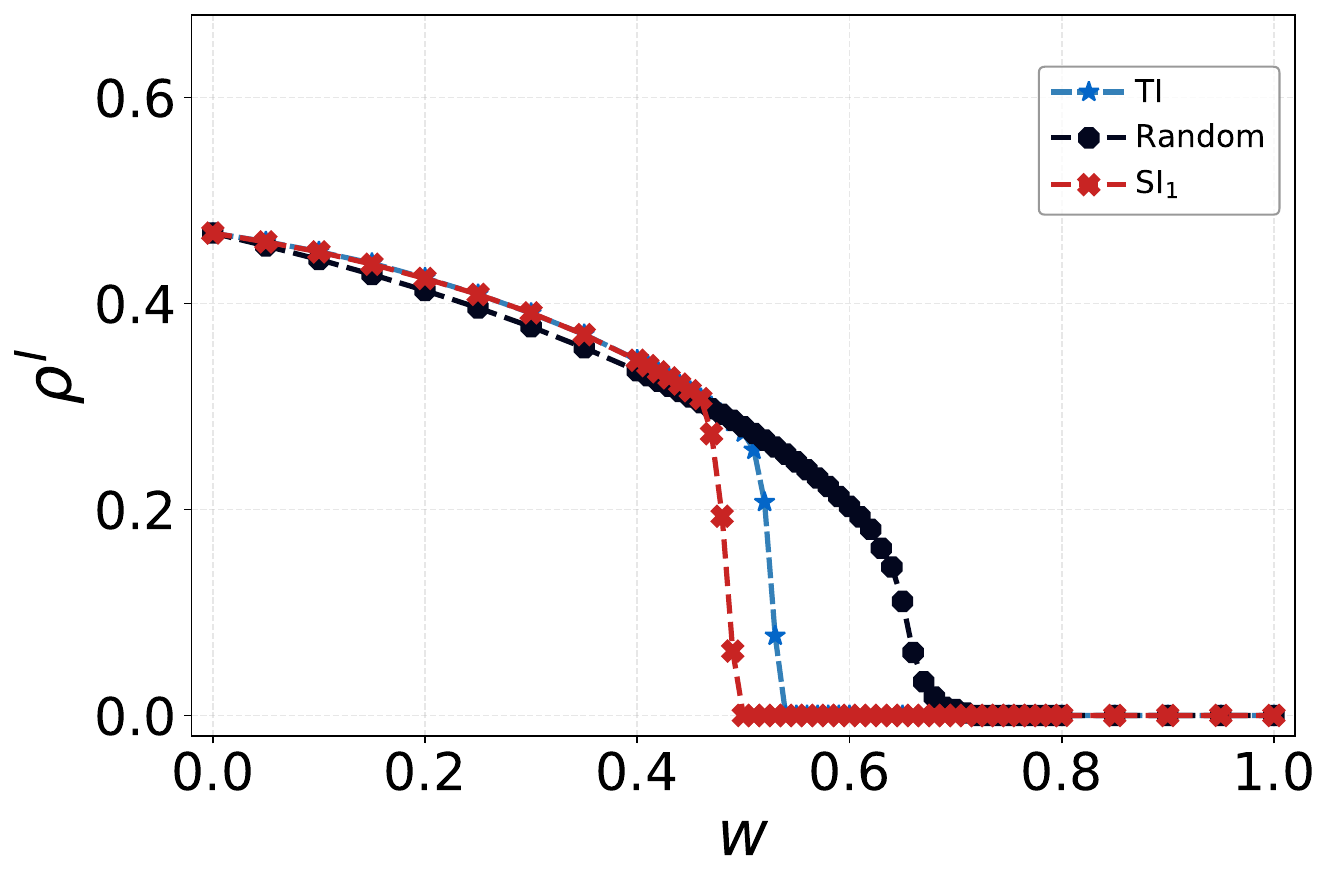}
        \caption{}
        \label{fig:real_b}
    \end{subfigure}
    \hfill
    \begin{subfigure}[b]{0.325\textwidth}
        \centering
        \includegraphics[width=\textwidth]{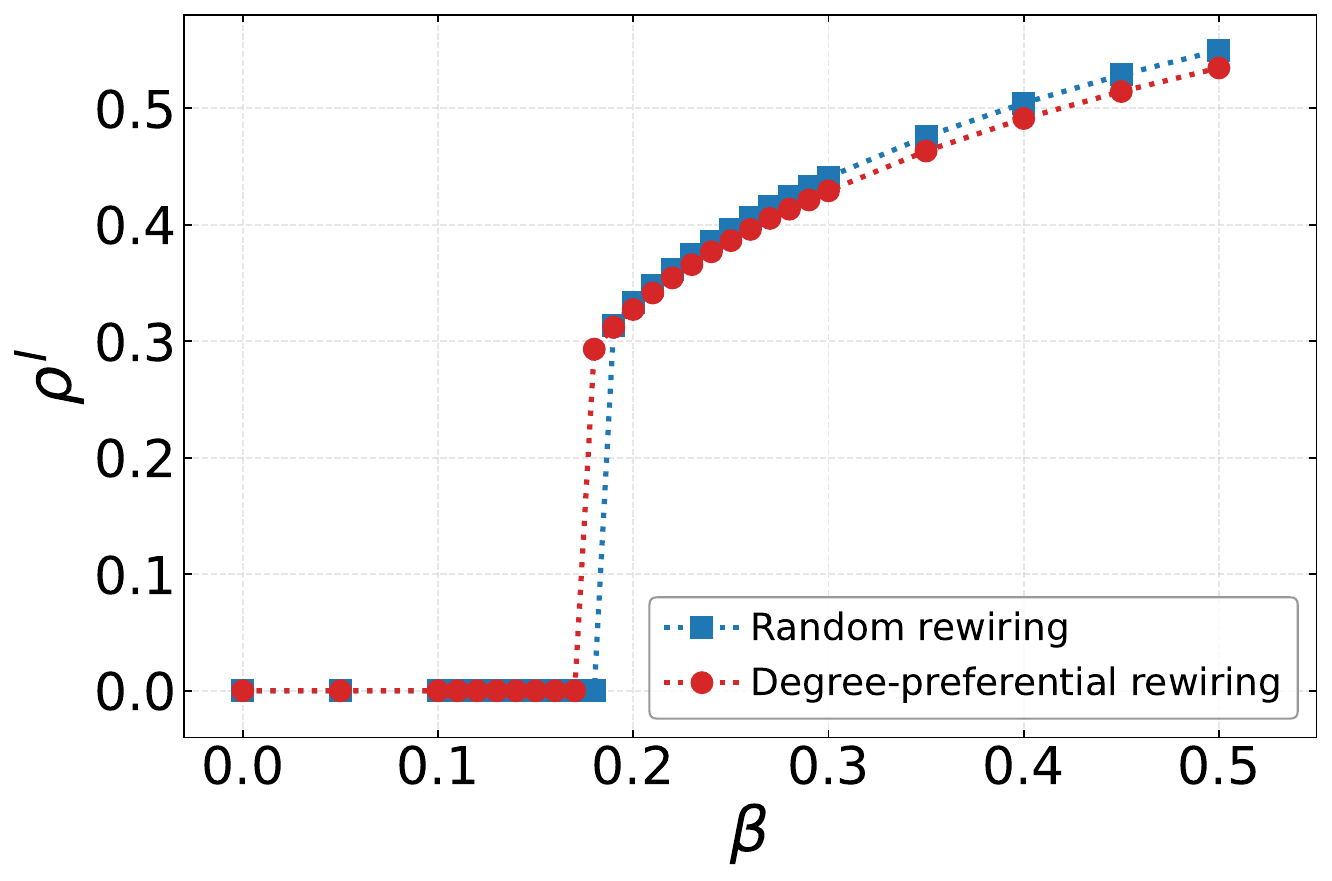}
        \caption{}
        \label{fig:real_c}
    \end{subfigure}
    
    \caption{
        \textbf{Performance of the proposed model and intervention strategies on the congressional bill cosponsorship hypergraph:}
        (a) Stationary infection density $\rho^I$ as a function of the infection rate $\beta$ under different sensitivity parameters $\eta$.
        (b) Comparison of different immunization strategies among targeted immunization (TI), random edge immunization (Random), and spontaneous isolation with threshold $\theta_{\min}=0.2$ (SI$_1$).
        (c) Comparison of stationary infection density under random rewiring and degree-preferential rewiring mechanisms after targeted immunization.
       All simulations use $\gamma = 0.03$, $\mu = 0.1$, and initial infection density $I(0) = 0.3$.}
    \label{fig_real_data}
\end{figure*}

To assess the applicability of the proposed model to real-world settings with heterogeneous group sizes, we further test it on an empirical hypergraph constructed from congressional bill cosponsorship data~\cite{James}.
The dataset contains 1218 nodes, each representing a US congressperson, and 8736 hyperedges corresponding to bills jointly cosponsored in a given year.
The hyperedges vary substantially in size, with an average of 10.4 cosponsors per bill and the largest bill involving 52 cosponsors.
Each hyperedge thus captures a natural group interaction among the legislators who collaborated on a bill, producing a network with widely varying interaction sizes.

The results, summarized in Fig.~\ref{fig_real_data}, confirm that the proposed model and intervention strategies remain effective beyond the synthetic uniform hypergraph setting.
As shown in Fig.~\ref{fig_real_data}(a), consistent with the findings on synthetic hypergraphs, higher values of the sensitivity parameter \(\eta\) suppress the endemic state by rapidly reducing hyperedge activity in response to elevated infection pressure, thereby limiting group interactions and expanding the disease‑free region.
The comparison of immunization strategies in Fig.~\ref{fig_real_data}(b) reveals that spontaneous isolation (SI\(_1\)) remains highly effective on the general hypergraph with heterogeneous edge sizes, outperforming targeted immunization (TI) and showing a substantial margin over random edge immunization.
This indicates that even in realistic higher‑order contact patterns, locally triggered deactivation of highly infected groups can serve as a robust containment mechanism.
The contrast between the two post‑intervention rewiring strategies, represented in Fig.~\ref{fig_real_data}(c), shows that degree‑preferential rewiring achieves a modest but consistent reduction in stationary prevalence relative to random rewiring.
Under the current setup, new hyperedges are generated with a size equal to the mean edge size of the original network, which is approximately \(10\).
The presence of isolated nodes and small hyperedges of size \(2\) or \(3\) in the original network constrains the structural advantage of degree‑preferential rewiring, as these low‑degree nodes provide limited opportunities for the emergence of the heavy‑tailed degree distributions that drive adaptive suppression.
This observation suggests that further refinements of the rewiring rules, possibly incorporating infection‑aware or size-dependent mechanisms, may be needed for general hypergraphs.

\section{Conclusion}		
\label{sec:conclusion}

In this work, we establish an adaptive framework for epidemic dynamics on adaptive hypergraphs by integrating simplicial SIS processes with adaptive hyperedge activity, supported by MMCA analysis and quasistationary Monte Carlo simulations. Our results demonstrate that higher-order interactions fundamentally reshape epidemic spreading through adaptive feedback, giving rise to discontinuous phase transitions, bistability, and hysteresis. The theoretical analysis further reveals nonlinear epidemic thresholds in genuinely higher-order settings, where the disease-free state remains stable against infinitesimal perturbations and endemic states emerge only via saddle-node bifurcations.
Moreover, intervention strategies that exploit the adaptive structure exhibit marked improvements in containment performance. Targeted immunization based on infection-pressure ranking and low-threshold spontaneous isolation consistently outperforms random baselines, while degree-preferential rewiring achieves long-term suppression by inducing heterogeneous topological patterns that reinforce adaptive feedback. In contrast, random rewiring provides only short-lived benefits due to the rapid restoration of transmission pathways.

Despite these advances, the present framework assumes uniform adaptive parameters and a single-layer representation, overlooking real-world heterogeneity.
In particular, all hyperedges share the same sensitivity and recovery rates, whereas actual groups likely differ in their responsiveness, potentially yielding richer dynamics such as coexisting endemic and disease-free subgroups.
However, since the IBMF framework tracks node probabilities independently, our model can naturally be extended to more general structures---for example, by introducing size-dependent transmission rates $\beta_d$ and a hyperdegree distribution $P(k_d)$ for hypergraphs with heterogeneous hyperedge sizes, or by accounting for all infection-triggering subsets within simplices in simplicial complexes.
Such extensions introduce additional nonlinear terms while preserving the overall MMCA structure.

In summary, our findings highlight the central role of higher-order interactions, adaptive connectivity, and targeted interventions in shaping epidemic dynamics and enhancing controllability on adaptive hypergraphs. Beyond the macroscopic behaviors uncovered in our analysis, the simulations show that highly connected nodes accumulate disproportionately high infection pressure, which in turn triggers stronger adaptive suppression of their incident hyperedges. This localized suppression cascades through the system and significantly amplifies the effectiveness of interventions that specifically target structurally vulnerable or behaviorally responsive groups. Building on these insights, future work could incorporate information-driven decision rules in which awareness, perceived risk, or digital signals dynamically modulate hyperedge activity or integrate optimal intervention scheduling under limited resources.
Such developments would deepen the theoretical understanding of adaptive contagion and support the design of more effective public health policies in complex social systems.

\singlespacing
\setlength\bibsep{0pt}

\end{document}